\documentclass[a4paper,11pt]{article}
\usepackage[utf8]{inputenc}

\usepackage{amsthm,amssymb,amsfonts,amsmath}
\usepackage{amssymb,amsfonts,amsmath}
\usepackage{xcolor}
\newtheorem{theorem}{Theorem}[section]
\newtheorem{lemma}[theorem]{Lemma}
\newtheorem{corollary}[theorem]{Corollary}
\newtheorem{proposition}[theorem]{Proposition}
\newtheorem{observation}[theorem]{Observation}

\newtheorem{problem}[theorem]{Problem}

\newtheorem{innercustomthm}[theorem]{Theorem}
\newenvironment{customthm}[1]
  {\renewcommand\theinnercustomthm{#1}\innercustomthm}
  {\endinnercustomthm}

\usepackage{graphicx}
\graphicspath{{.}}
\usepackage{wrapfig}

\usepackage{fullpage}

\newcommand{\CH}{\ensuremath{\mathrm{CH}}}
\newcommand{\seq}[1]{\ensuremath{\left \langle #1 \right \rangle}}

\newcommand{\etal}{et al.}

\newcommand{\fig}[1]{\figurename~\ref{#1}}

\makeatletter
\g@addto@macro\bfseries{\boldmath}
\makeatother

\usepackage{authblk}

\title{An Optimal Algorithm for Reconstructing Point Set Order 
Types from Radial Orderings}
\author[1]{Oswin Aichholzer}
\author[2]{Vincent Kusters}
\author[3]{Wolfgang Mulzer}
\author[2]{Alexander Pilz}
\author[2]{Manuel Wettstein}

\affil[1]{Institute for Software Technology, Graz University of Technology. \authorcr \texttt{oaich@ist.tugraz.at}.
}
\affil[2]{Department of Computer Science, ETH Z\"urich. \authorcr
\texttt{[vincent.kusters, alexander.pilz, manuelwe]@inf.ethz.ch}.
}
\affil[3]{Institut f\"ur Informatik, Freie Universit\"at Berlin. \authorcr \texttt{mulzer@inf.fu-berlin.de}.
}

\begin{document}

\maketitle

\begin{abstract}
Let $P$ be a set of $n$ labeled points in the plane. 
The \emph{radial system} of~$P$ describes, for each 
$p\in P$, the order in which a ray that rotates
around $p$ encounters the points in $P \setminus \{p\}$. 
This notion is related to the \emph{order type} 
of~$P$, which describes the orientation (clockwise 
or counterclockwise) of every ordered triple in~$P$.
Given only the order type, the radial
system is uniquely determined and can 
easily be obtained. The converse, however,
is not true.  Indeed, let $R$ be the radial system
of $P$, and let $T(R)$ be the set of all order
types with radial system $R$
(we define $T(R) = \emptyset$ for the
case that $R$ is not a valid radial system).
Aichholzer \etal~(\emph{Reconstructing 
Point Set Order Types from Radial Orderings}, in 
Proc.~ISAAC 2014) show that 
$T(R)$ may contain up to $n-1$ order types. 
They also provide polynomial-time algorithms to 
compute $T(R)$ when only $R$ is given.

We describe a new algorithm for 
finding $T(R)$.  The algorithm constructs the convex 
hulls of all possible point sets with the radial 
system $R$. After that, orientation queries on point triples 
can be answered in constant time. A representation of this set of convex 
hulls can be found in $O(n)$ queries to the radial system,
using $O(n)$ additional processing time. This is optimal. 
Our results also generalize 
to \emph{abstract order types}.
\end{abstract}

\section{Introduction}\label{sec:intro}

Let $P = \{p_1, \dots, p_n\}$ be a set of $n$ labeled points
in the plane, where the point $p_i$ is given the label $i$.
The \emph{chirotope} $\chi : [n]^3 \rightarrow \{-1, +1, 0\}$
of~$P$ is a function that assigns to 
each triple $(i, j, k) \in [n]^3$
the orientation $\chi(i, j, k)$
 of the corresponding point triple $(p_i, p_j, p_k) \in P^3$
(clockwise ($-1$), counterclockwise ($+1$), or collinear~($0$)). 
If the elements of $(i, j, k)$ are not pairwise distinct,
then we set $\chi(i, j, k) = 0$.
Throughout this paper, 
we assume that $P$ is in  \emph{general position}, i.e.,
its chirotope~$\chi$ has 
$\chi(i, j, k) \neq 0$, for all $(i,j,k) \in [n]^3$
with pairwise distinct elements.

Let $P$ and $P'$ be two sets of $n$ labeled points in the plane,
and let $\chi$ and $\chi'$ be their chirotopes.
We say that $\chi$ and $\chi'$ are \emph{equivalent}
if either $\chi(i, j, k) = \chi'(i, j, k)$, for all
$(i, j, k) \in [n]^3$, or 
$\chi(i, j, k) = -\chi'(i, j, k)$, for all
$(i, j, k) \in [n]^3$.
This defines an equivalence relation on the
chirotopes. An equivalence class in this relation is
called \emph{order type}.
Many problems 
on planar point 
sets do not depend on the exact coordinates of the points 
but only on their order type. Examples include computing 
the convex hull and determining whether two segments with
endpoints in the point set intersect.
As far as algorithms are concerned, it often turns out 
that access to the order type suffices in order to
obtain efficient results.
For example, Knuth\cite{knuth1992axioms} 
shows that the
convex hull of a point set can be computed in 
$O(n \log n)$ time, even if one can only access its order type.%
\footnote{Actually, Knuth considers the generalized setting of abstract order types (to be defined later);
many algorithms, as, e.g., Graham's scan, can also be slightly adapted to work by accessing only order type information.}

Given a function $\chi : [n]^3 \rightarrow \{-1, +1, 0\}$,
it is a hard problem to determine whether $\chi$
is a chirotope for a labeled planar point set.\footnote{To be
precise, this problem is complete for the existential
theory of the reals $\exists\mathbb{R}$.\cite{Matousek14}} To get
around this difficulty, one uses the notion of \emph{abstract
order types}. Recall that an arrangement of \emph{pseudo-lines}
in the plane is a set of $x$-monotone planar curves
such that each pair of curves intersects in exactly
one point and such that this intersection is crossing.
A \emph{generalized configuration of points} consists of a labeled
point set $P$ in the plane and an arrangement of pseudo-lines such that 
each pseudo-line contains exactly two points from $P$ and such
that each pair of points from $P$ lies on a 
pseudo-line.\cite{goodman_proof}
Now we can define a chirotope on $P$ as follows:
if a point $p_k \in P$ is to the left 
of the pseudo-line through $p_i, p_j \in P$, directed from 
$p_i$ to $p_j$, then the triple $(i,j,k)$ is oriented 
\emph{counterclockwise}. Otherwise, it is oriented \emph{clockwise}. 
An equivalence class of chirotopes obtained in this way is
called an \emph{abstract order type}. 
Abstract order types can be characterized by a simple set
of axioms.\cite{knuth1992axioms}
For most combinatorial 
purposes, generalized configurations of points behave like 
point sets: their convex hull is the intersection of 
those halfspaces bounded by the pseudolines that contain 
all the points, and it determines a cycle of directed arcs.
The chirotope determines whether two arcs defined by 
pairs of points cross. We refer to the work of Goodman 
and Pollack (see, e.g., their work on semispaces of 
configurations\cite{semispaces}) and to a book by 
Knuth\cite{knuth1992axioms} (who calls abstract order 
types ``CC systems'') for more details.

In this paper, we will be solely concerned with abstract 
order types. We stress that, as opposed to many other publications on 
the subject, we consider \emph{labeled} 
abstract order types (and do not consider chirotopes
equivalent if they can be obtained by a permutation of their
arguments). In the following, we will not distinguish
between an abstract order type and a chirotope that
represents it.

\paragraph{Radial systems.}
We now define the main notion studied in this paper.
Let $P = \{p_1, \dots, p_n\}$ be a generalized configuration of points,
and let $\chi$ be the abstract order type of $P$.
The \emph{counterclockwise radial system} of 
$\chi$, denoted $R_\chi$, assigns to each
$i \in [n]$ the cyclic permutation $R_\chi(i)$ of 
$[n] \setminus\{i\}$ that is given
by the labels of the points in $P \setminus \{p_i\}$
in counterclockwise order around $p_i$. We call each 
$R_\chi(i)$ a \emph{counterclockwise radial ordering}.
If $\chi$ is realizable as a point set, then $R_\chi(i)$ equals the order of point labels 
found by sweeping a ray around $p_i$ in counterclockwise 
direction. Given a function $U$ that assigns to
each $i \in [n]$ a cyclic permutation $U(i)$ of $[n] \setminus \{i\}$,
we write $U \sim R_\chi$ 
if, for all $i \in [n]$, it holds that $U(i)$ equals
$R_\chi(i)$ or the reverse of $R_\chi(i)$.
Thus, the relation $\sim$ ``forgets'' the 
clockwise/counterclockwise direction of each individual~$R_\chi(i)$. We call an equivalence class under $\sim$  an 
\emph{undirected radial system}.
When we say \emph{radial system}, we always mean counterclockwise 
radial system. 
Radial systems were studied systematically by 
Aichholzer et al.\cite{urs} Before we describe
their results, let us first review some related notions
that have appeared in the literature.

\paragraph{Related work.}
Variants of the notion of radial systems have 
been studied in many contexts. First and foremost, 
there is the concept of 
\emph{local sequences}. Whereas our 
radial orderings are obtained by sweeping a ray around
each point, local sequences are obtained by sweeping 
a line. More precisely, let $P$ be a finite point set
in the plane. For a point~$p \in P$,
the \emph{local sequence 
of ordered switches} of $p$ is the cyclic sequence in 
which the points of~$P$ are encountered when rotating 
a directed line through~$p$. Additionally, we record 
whether a point appears 
before or after~$p$ on the directed line. Without
this additional information, we get 
the \emph{local sequence of unordered switches}. 
Goodman and Pollack\cite{semispaces} show that 
both concepts determine the order type of~$P$, 
and thus carry the same information. Wismath\cite{wismath} 
describes a method to reconstruct a point set (up 
to vertical translation and scaling) from its 
local sequences of ordered switches if, in addition, 
the $x$-coordinates of the points and the local 
sequences of directed switches are given.
He also mentions that the radial system does not 
always determine the order type.
Felsner and Weil\cite{felsner_weil} (Theorem~8) and 
Streinu\cite{clusters_of_stars} independently 
obtain a necessary and sufficient condition 
for sequences to be local sequences of unordered 
switches of an abstract order type. This condition allows 
for testing their realizability in polynomial time.

Another variation on radial systems was studied by 
Tovar, Freda, and LaValle\cite{TFL07} in the 
context of a robot that can sense landmarks around it.
Disser et al.\cite{DMW10} and Chen and Wang\cite{CW12} 
consider the \emph{polygon reconstruction problem 
from angles}, where the objective is to reconstruct 
a polygon when given, for each vertex $v$, the 
angles with the other vertices of the polygon visible
from $v$. Pilz and Welzl\cite{ooo} describe 
a hierarchy on order types based on crossing edges;
two order types are equivalent in their partial order 
if and only if they have the same radial system.
We refer to the work by Aichholzer et al.\cite{urs} for 
a more complete list of related work.

\paragraph{Good drawings.}
Radial systems are also closely related
to \emph{good drawings}. Let $G$ be
a graph.
A \emph{drawing} of $G$ is a representation of $G$ with vertices as distinct points in the plane or on the sphere and edges as Jordan arcs whose endpoints are the corresponding vertices.
It is usually assumed that no edges pass through vertices and two edges intersect only in a finite number of points.
A \emph{good drawing} (sometimes also 
called a \emph{simple topological graph}) of 
$G$ is a drawing of $G$ in the plane or on the 
sphere where each vertex is represented by 
a distinct point, and each edge is represented 
by a Jordan arc between its two vertices;
any two such arcs intersect in at most one 
point, which is either a common endpoint 
or a proper crossing.
Two drawings on the sphere are \emph{isomorphic} if they are equivalent under a homeomorphism of the sphere.
We consider two drawings in the plane isomorphic if they are isomorphic after a stereographic projection to the sphere.

The \emph{rotation} of a vertex $v$ in a drawing is the cyclic order of the 
edges incident to $v$. The \emph{rotation 
system} of a drawing is the set of 
the rotations of its vertices.
Clearly, good drawings are a generalization of geometric graphs.
The radial 
system of a point set $P$ is equivalent to 
the rotation system of the complete geometric 
graph on $P$. A generalized configuration of 
points $Q$ defines a good drawing of $K_n$ 
where the vertices are embedded on the points 
of $Q$ and every edge is a segment of a pseudo-line
in $Q$.

A good drawing is \emph{pseudo-linear} 
if its edges can be simultaneously extended 
to obtain a pseudo-line arrangement, or if it is isomorphic to such a drawing.
The good 
drawings obtained from generalized configurations
of points are exactly the pseudo-linear drawings (up to isomorphism).
See \figurename~\ref{fig_good_drawing_examples} for examples.
The radial system of $Q$ is equivalent to the 
rotation system of this good drawing. In a good 
drawing of~$K_n$, the rotation system determines 
which edges cross. Therefore, it fixes the drawing 
up to the ordering of the crossings; in particular, 
we can see whether two edges cross by locally 
inspecting the rotations for the four vertices 
involved.\cite{kyncl_realizability} Later, we will use 
good drawings as an important tool in our 
reconstruction algorithm.

It is well-known that not every rotation system can be realized by a good drawing of the corresponding graph.
Kyn\v{c}l\cite{kyncl_simplified} showed that a rotation system of $K_n$ is the rotation system of a good drawing if and only if this is true for every 5-vertex subset.
He approaches the problem from the aspect of \emph{abstract topological graphs}, where a graph is given together with a list of crossing edge pairs.
For abstract topological graphs of $K_n$, he shows that from every 6-vertex subset one can obtain the unique rotation system of the corresponding good drawing, if it exists.
For non-complete abstract topological graphs, the realizability problem is NP-complete.\cite{kyncl_realizability}
Deciding whether there is a good drawing of a non-complete graph with a given rotation system seems to be an open problem.
(There, the rotation system no longer determines the set of crossing edge pairs.)

Some good drawings are isomorphic to drawings where each edge is an $x$-monotone curve (after a projection to the plane).
Such drawings are called \emph{monotone}.
Clearly, all pseudo-linear drawings are monotone.
(Using Lemma~\ref{lem:important_triangles_not_crossed}, it is an easy exercise to provide an example showing that the converse is not true;
Kyn\v{c}l\cite{kyncl_realizability} provides all five non-isomorphic good drawings of $K_5$, of which only three are pseudo-linear.)
Balko, Fulek, and Kyn\v{c}l\cite{monotone} characterize monotone good drawings of $K_n$, and Aichholzer et al.\cite{monotonicity_algorithm} provide an $O(n^5)$ time algorithm for deciding whether a given rotation system is the one of a monotone good drawing of~$K_n$.
For non-complete graphs, no similar results are known.

In terms of rotation systems of good drawings, our algorithm solves the problem for pseudo-linear drawings of~$K_n$,
that is, whether a given rotation system is the one of a pseudo-linear drawing of~$K_n$.
We are not aware of any related results in connection with non-complete graphs.
Note that our problem is not concerned with finding a good drawing of a given rotation system, but with, in these terms, deciding whether the rotation system is the one of a pseudo-linear drawing, and determining the edges of all possible unbounded cells in all possible pseudo-linear drawings.

Also note that for any good drawing, the vertex triples can be oriented by defining the unbounded cell.
However, this must not be confused with the order type, as only point sets have an order type.
Finally, let us recall that there are good drawings of $K_n$ that have the same rotation system, but are non-isomorphic to each other.
In particular, even though the set of crossing edge pairs is determined, as well as the direction in which an edge crosses another one (the \emph{extended rotation system}, the order in which an edge is crossed by other edges is, in general, not fixed).
(This is not even the case for geometric graphs; e.g., slightly perturbing the vertices of an almost-regular hexagon influences the order in which its diagonals cross.)
A detailed discussion of this can also be found in Kyn\v{c}l\cite{kyncl_realizability}.
However, due to a result by Gioan\cite{gioan}, if a drawing of $K_n$ is pseudo-linear, then all good drawings with the same rotation system are pseudo-linear as well.

\begin{figure}
\centering
\includegraphics{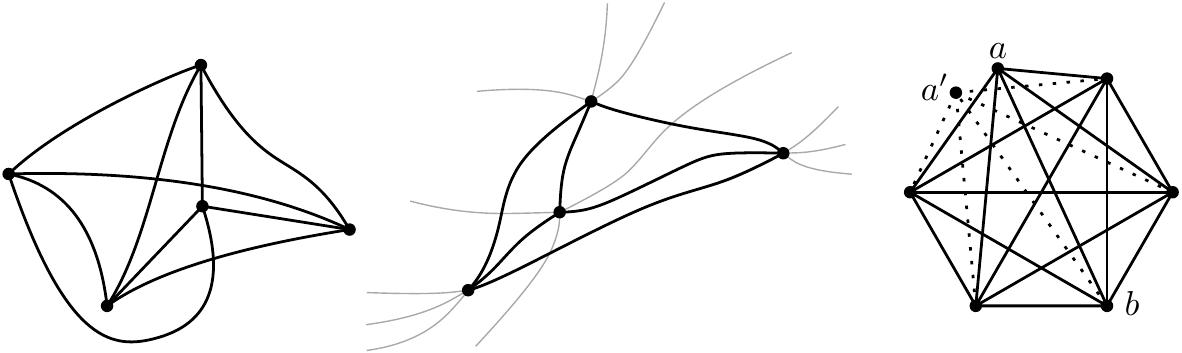}
\caption{Left: A good drawing of $K_5$.
It is monotone since there is a homeomorphism of the plane s.t.\ all edges are $x$-monotone.
Middle: A generalized configuration of points with gray pseudo-lines.
The induced pseudo-linear good drawing is shown by the black edges.
Right: A geometric drawing of $K_6$ and a perturbation of vertex $a$
to vertex $a'$, that shows that, even though the original and the
perturbed drawing have the same rotation system (i.e., the underlying
point set has the same radial system), and the directions in which two
edges cross each other match, the order in which other edges cross
$ab$ is different from the one of $a'b$ (giving two non-isomorphic labeled drawings).
}
\label{fig_good_drawing_examples}
\end{figure}

\paragraph{Properties of Radial Systems.}
Before we can describe our results, we
provide a quick overview of the 
previous work.
Aichholzer et al.\cite{urs} investigated 
under which circumstances the undirected radial system $U$
of a generalized configuration of points $P$ uniquely determines 
the abstract order type $\chi$. They show that if $P$ has 
a convex hull with at least four points, then $U$ uniquely 
determines~$\chi$.
In the following, let $U$ be an
undirected radial system that originates from an
abstract order type, and let  $T(U)$ be 
the set of abstract order types with undirected radial 
system~$U$.

\begin{theorem}[Theorem~1 and 2 in Aichholzer et al.\cite{urs}]
\label{thm:urs_combined}
Let $n \geq 5$ and 
consider an abstract order type $\chi$ on $[n]$.
Let $U$ be the undirected
radial system of $\chi$, and let 
$H\subseteq [n]$ be the elements of the convex hull of 
$\chi$. Then, we can compute $|H|$ from $U$ in polynomial 
time. Furthermore,
\begin{enumerate}
\item[(i)] if $|H|\neq 3$, then $T(U)=\{\chi\}$, and 
we can compute $\chi$ from $U$ in polynomial time; and
\item[(ii)] if $|H|=3$, then $|T(U)|\leq n-1$; all 
elements of $T(U)$ have a convex hull with exactly three
elements; and 
we can compute $T(U)$ from $U$ in polynomial time.
\end{enumerate}
\end{theorem}

In the full version of their paper\cite{urs}, 
Aichholzer et al. show that (i) can be implemented 
in $O(n^3)$ time. Furthermore, they show that there exist counterclockwise 
radial systems $R$ with $|T(R)|=n-1$. Hence, it 
is not possible to improve the bound on $|T(U)|$ 
in (ii), even if we consider counterclockwise 
radial systems instead of undirected radial systems.\cite{urs}

Although $U$ does not always uniquely determine $\chi$, 
the pair $(U,H)$, where $H$ is the set of elements 
on the convex hull of $\chi$, always suffices.\cite{urs}
Thus, the abstract order types in $T(U)$ all have 
different convex hulls. Given an undirected radial 
system $U$ on $[n]$, we say that a subset 
$H\subseteq [n]$ is \emph{important} if $H$ is 
the convex hull of some abstract order type in $T(U)$.
An \emph{important triangle} is an important 
set of size $3$. Important sets are interrelated as follows.

\begin{theorem}[Propositions~1--4 in Aichholzer et al.\cite{urs}]
\label{thm:important_triangles_structure}
Let $n \geq 5$ and
consider a radial system $R$ on $[n]$.
If $R$ has more than two important triangles, 
then all important triangles must 
have one element $i^*\in [n]$ in common.
Thus, combining with Theorem~\ref{thm:urs_combined},
we can conclude that exactly one of the 
following cases applies:
\begin{enumerate}
  \item[(1)] There is an important set of
    size at least four, which is the only important set.
   \item[(2)] There are between $1$ and $n-1$ important sets.
   All important sets are triangles, and if there is more than
   one important set, there is an element $i^* \in [n]$
   that is contained in all of them.
  \item[(3)] There are exactly two important 
    sets, they are triangles, and they are disjoint.
  \end{enumerate}
\end{theorem}

For cases (2) and (3), there is actually a 
complete characterization of the important triangles.
For an abstract order type $\chi \in T(U)$, an 
\emph{inner} important triangle of $\chi$ is an 
important triangle of $U$ that is not equal to 
the convex hull of~$\chi$. The following lemma 
reformulates the fact that an inner important 
triangle is not contained in a convex quadrilateral 
(see \fig{fig_important_partition_a}).

\begin{figure}
\centering
\includegraphics{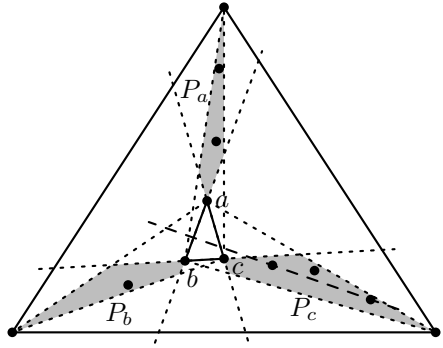}
\caption{An inner important triangle $\seq{a,b,c}$ partitions the 
point set into three subsets, in which each pair of 
points intersects the opposite edge.}
\label{fig_important_partition_a}
\end{figure}
\begin{lemma}[Aichholzer et al.\cite{urs}, Pilz and Welzl\cite{ooo}]
\label{lem:important_partition}
Let $P$ be a generalized configuration of $n$ points, and
let $\chi$ be the abstract order type of $P$.
A triple~$(a,b,c) \in [n]^3$ is an inner important 
triangle iff the following conditions hold.
\begin{enumerate}
\item[(1)] The triangle $p_ap_bp_c$ is empty of points of~$P$.
\item[(2)] The triangle $p_ap_bp_c$ partitions $P \setminus \{p_a,
p_b, p_c\}$ into three subsets $P_a$, $P_b$, and $P_c$, 
such that $P_a$ is to the left of the directed 
line $p_bp_a$ and to the right of the directed line $p_cp_a$, 
and similarly for $P_b$ and $P_c$. 
\item[(3)] For any two points $v, w \in P_a$, 
  the pseudo-line $vw$ intersects the edge~$p_bp_c$;
  and similarly for $P_b$ and $P_c$.
\end{enumerate}
\end{lemma}

In this context, we mention that, if $R$ 
is the radial system of some point set order 
type, then every abstract order type with radial 
system $R$ can be realized as a point 
set (see Theorem~27 in Pilz and Welzl\cite{ooo}).
We do not deal with the realizability of abstract 
order types as point sets in this work.
In the following, by a \emph{realization of 
a radial system} $R$, we mean an abstract order
type whose radial system is $R$.

\paragraph{Our results.}
We assume that the radial system
is given in a data structure that 
lets us obtain the relative order of three 
elements in a radial ordering in constant time.
We call such a query a \emph{triple test}.
For example, a triple test can be carried out
in $O(1)$ time if we store not only the radial 
ordering, but also the rank of each element within 
some linear order defined by the radial ordering 
around each vertex.
(If this structure is not provided, it can be obtained in $\Theta(n^2)$ time.)

For a given undirected radial system $U$ 
on $n$ elements (which has size~$\Theta(n^2)$), 
we provide an algorithm to
direct the $n$ radial orderings in a
consistent manner in $O(n)$ triple tests
and $O(n)$ additional time
(Theorem~\ref{thm:directing_radial_systems}).

Our main algorithm identifies the convex 
hulls of all abstract order types consistent with 
a given radial system $R$ in $O(n)$ time (provided 
that $R$ is the radial system of an abstract order type).
This set allows for constant-time queries to 
each chirotope in $T(R)$. Throughout the paper,
when we speak of the ``convex hull'' of a
set of vertices, we mean a combinatorial
representation as the cyclic
permutation of the vertices that
appear on the convex hull, in this order.

\begin{theorem}\label{thm:linear_preprocessing}
Given a radial system~$R$ of an abstract order type, 
we can find in $O(n)$ triple tests and $O(n)$
additional processing a data structure that
represents the convex hulls of all chirotopes in $T(R)$.
With this data structure, 
we can answer queries to the chirotopes of $T(R)$ in 
constant time. 
\end{theorem}

Hence, this is a means of reporting an explicit 
representation of $T(R)$ in $O(n)$ time, significantly 
improving Theorem~\ref{thm:urs_combined}.
We remark that we can show that $\Omega(n)$ triple tests
are necessary, as an 
adversary can use any unconsidered point in a suitable 
example to alter $|T(R)|$ (e.g., by using it to 
``destroy'' a top triangle as defined in 
Section~\ref{sec:obtaining_hull_edges}, see 
Proposition~\ref{prop:number_lower_bound}). 
In this sense, our result is optimal.

We assume that the input consists of the
permutations of an actual radial system.
If this is not the case, our algorithm might
fail, because it operates under an assumption that
is not satisfied, or it may compute a structure
that represents all chirotopes that are consistent with the
triple tests performed by the algorithm.
%
If we do not know that the set of permutations provided 
as input is indeed the radial system of an abstract 
order type, we show how to verify this in $O(n^2)$ time.
For some input~$R$, we define $T(R) = \emptyset$ if $R$
is not a valid radial system.
A straight-forward adversary argument shows that $\Omega(n^2)$ 
triple tests (i.e., reading practically the whole input) is 
necessary to verify whether $T(R)= \emptyset$.
(The adversary can exchange two unread elements in the 
radial ordering around a point, cf. 
Proposition~\ref{prop:existence_lower_bound}.) 

Finally, 
when considering the algorithmic complexity of 
determining the realizability of a radial system, 
the question arises whether there are constant-size 
non-realizable subsets in any non-realizable radial system.
If this were the case, one could check for realizability
by examining the induced radial systems up to a
certain constant size. Unfortunately, this
is not the case.
In Section~\ref{sec:minimal_unrealizable}, we show the 
following result.

\begin{theorem}
\label{thm:minimal_unrealizable}
For any $k \geq 3$, there exists a radial 
system $R_k$ over $n = 2k + 1$ elements that 
is not realizable as an abstract order type, 
but that becomes realizable as a point set 
order type when removing any point.
\end{theorem}

\section{Directing Undirected Radial Systems}
Let $U$ be an undirected radial system on $n$ elements.
We show how to obtain a counterclockwise radial system 
$R$ from $U$ in $O(n)$ triple queries. Thus, for the remaining 
sections, we can assume that any radial system is oriented 
counterclockwise. 
We use two simple observations for $4$-sets of points of 
a set with counterclockwise radial system~$R$.
A \emph{swap} for an element $i$ inverts the radial ordering $R(i)$ 
(resulting in a new radial system).
Note that when restricting $R(i)$ to a 4-set containing $i$, since there 
are only two ways in which three vertices can be ordered around a fourth 
one, a swap corresponds to changing the order of two (arbitrary)
neighboring elements in $R(i)$ restricted to this 4-set.
We use two crucial observations.

\begin{observation}\label{obs_4tuple_p4}
For a $4$-set $\{i, j, k, l\}$, the 
counterclockwise radial ordering $R(l)$ is 
uniquely determined by the counterclockwise 
radial orderings $R(i)$, $R(j)$, and $R(k)$.
\end{observation}
\begin{observation}\label{obs_4tuple_swaps}
For a $4$-set $\{i, j, k, l\}$, consider two counterclockwise radial 
systems $R_A$ and $R_B$ that are realized by abstract order types.
Then the counterclockwise radial orderings of $R_A(i)$, $R_A(j)$,
$R_A(k)$, and $R_A(l)$ differ from the counterclockwise radial 
orderings of $R_B(i)$, $R_B(j)$, $R_B(k)$, and $R_B(l)$ 
by an even number of swaps.
\end{observation}

\begin{theorem}\label{thm:directing_radial_systems}
Let $n \geq 5$, 
let $\chi$ be an abstract order type on $n$ elements, and 
let $U\sim R_\chi$.
Then $U$ uniquely determines $R_\chi$ (up to complete reversal), 
and we can compute $R_\chi$ from $U$ in $O(n)$ triple queries
and $O(n)$ total time.
\end{theorem}
\begin{proof}
We choose the direction of $R(1)$ arbitrarily.
Now, there are four possible choices for the directions of 
$R(2)$ and $R(3)$.  For each choice, we consider the resulting 
induced counterclockwise radial system on $\{1, \dots, 5\}$, and
we  check whether it is realizable as an abstract order type.
This can be done in constant time, as there is only a constant 
number of abstract order types on five elements.
If no choice yields a realizable counterclockwise radial system,
$U$ cannot be realized, and we stop.

Next, we argue that at most one choice can lead to a realizable
counterclockwise radial system.
For the sake of contradiction, suppose that two different 
choices for $R(2)$ and $R(3)$ lead to 
realizable counterclockwise radial systems $R_A$ and $R_B$
on $\{1, \dots, 5\}$. 
Let us assume first 
that only $R(2)$ is inverted in~$R_B$.
Then, by applying Observation~\ref{obs_4tuple_swaps} on 
$\{1,2,3,4\}$ and on $\{1,2,3,5\}$, we see that also 
$R(4)$ and $R(5)$ are inverted.
But then for the $4$-set $\{1,2,4,5\}$, we have 
three swaps between $R_A$ and $R_B$, a 
contradiction to Observation~\ref{obs_4tuple_swaps}.
The same argument rules out that only $R(3)$ is inverted, so assume 
that both $R(2)$ and $R(3)$ are inverted.
Then, again by Observation~\ref{obs_4tuple_swaps}, 
the directions for $R(4)$ and $R(5)$ remain unchanged.
So, for the  $4$-set $\{1,2,4,5\}$, we have only one 
swap between $R_A$ and $R_B$, which is a contradiction.

Since the directions of $\{1, \dots, 5\}$ are now fixed, we use 
Observation~\ref{obs_4tuple_p4} to fix the direction of 
all $i = 6, \dots, n$ by considering the 
$4$-set $\{1, 2, 3, i\}$.
Thus, we conclude that if $U$ is realized by an 
abstract order type, we can obtain the 
unique counterclockwise radial system $R$ in~$O(n)$ time.
\end{proof}

Note that Theorem~\ref{thm:directing_radial_systems} actually holds 
for good drawings and not only for radial systems of abstract order types.

\section{Obtaining Chirotopes from Radial Systems}
\label{sec:obtaining_chirotopes_from_radial_systems}

Let $R$ be a given system of permutations. Our goal is to obtain the set 
$T(R)$ of abstract order types that realize $R$, if $R$ represents
a valid radial system.
Otherwise, $T(R)$ is empty.
Our algorithm (conceptually) constructs 
a good drawing of a plane graph on the sphere by adding vertices 
successively while maintaining the faces that are candidates for 
the convex hull. We will see that this actually boils down 
to maintaining at most two sequences of vertices plus one special vertex.
Throughout, we assume that the radial orderings 
indeed correspond to the radial system of an abstract order type.
If any of our assumptions does not hold, we know that there 
is no abstract order type for the given set of radial orderings.
If~$R$ can be realized as an abstract order type, then the plane 
graph is the subdrawing of a drawing weakly 
isomorphic~(cf.~Kyn\v{c}l\cite{kyncl_realizability}) to 
the complete graph on any generalized configuration of 
points that realizes that abstract order type.

Let $C=\seq{c_0, \dots, c_{m-1}}$ be a plane
cycle with $m$ vertices contained in a good drawing $\Gamma$ 
of the complete graph that realizes a radial system $R$.
(We think of $C$ as counterclockwise with the \emph{interior} 
to its left.) 
Let $c_i$ be a vertex of $C$, and let $v$
be a vertex not in $\{c_{i-1}, c_i, c_{i+1}\}$.\footnote{We 
consider all indices modulo the length of the corresponding sequence.}
We say that the edge $c_i v$ \emph{emanates to the outside} of 
$C$ if $v$ lies between $c_{i-1}$ and $c_{i+1}$
in $R(c_i)$ in counterclockwise order.
Otherwise, $c_i v$ \emph{emanates to the inside}.
Let $w$ be a vertex that does not belong to $C$.
If $c w$ emanates to the outside for all $c\in C$, 
then $w$ \emph{covers} $C$.
If $c w$ emanates to the inside for all $c \in C$, then $w$ 
\emph{lies inside} $C$; otherwise, $w$ \emph{lies outside} $C$.
If $w$ neither is inside $C$ nor covers $C$, then $\Gamma$
restricted to $C$ plus all edges from vertices of~$C$ to~$w$ is not plane.
A cycle $C=\seq{c_0, \dots , c_{m-1}}, m \geq 4$, in
$\Gamma$ is \emph{compact} 
if $C$ is plane and, for each $c_i$, the edges $c_i c_{i+2}, c_i c_{i+3}, 
\dots, c_i c_{i-2}$ all emanate to the inside.

\begin{observation}\label{obs:compact_cycle_convex}
Let $R$ be a radial system and let $C$ be a compact
cycle in $R$.
Let $P$ be a generalized configuration of points 
realizing $R$. Then the  vertices of $C$ are in 
convex position in $P$.
\end{observation}

\begin{lemma}\label{lem:important_triangles_not_crossed}
Consider a radial system~$R$,
and let $\Gamma$ be a good drawing of the complete graph 
whose rotation system corresponds to~$R$. Let $S$ be an important
set of $R$. 
Then, the vertices of $S$ define a cell in $\Gamma$, and 
no edge of $\Gamma$ crosses an edge of this cell.
Furthermore, no element of
$S$ lies inside a compact cycle in~$\Gamma$.
\end{lemma}
\begin{proof}
Let $P$ be a generalized configuration of points
that realizes $R$ and whose convex hull
$H_S$ is a cycle with vertex set $S$.
Let $G$ be the embedding of the complete graph
on $n$ vertices that is obtained by
taking the pseudoline arrangement for $P$
and keeping only the parts of the pseudolines
between the vertices.
Two edges in the embedding $G$ cross if and only if 
they cross in~$\Gamma$. In particular, $H_S$ is not crossed in $\Gamma$,
and there is no edge in $\Gamma$ that emanates to the outside of~$H_S$.

Now, suppose $\Gamma$ contains a compact cycle $C$ such 
that some vertices of $H_S$ are in the interior of $C$ ($C$ and $H_S$ may share some vertices).
\fig{fig_important_triangles_not_crossed} shows an example where $C$ and $H_S$ share two vertices.
Recall that we defined the ``interior'' of $C$ via its direction, so a point inside $C$ does not need to be drawn 
such that it is separated from the unbounded cell by~$C$.
We remove all vertices not in $S$ or $C$ from~$G$.
Note that the curves forming the diagonals of $C$ cross in a bounded
cell of the resulting embedded graph that is in the interior of both $H_S$ and $C$, as all elements of the graph are in the interior of $H_S$ and the crossings are in the interior of $C$ by assumption.
By removing some vertices from $H_S$ that are not vertices of $C$, the cell outside of $H_S$ (the unbounded face 
of $G$) grows.
Eventually, the boundary of this cell contains a crossing~$p$ from the interior of $C$, as $p$ is not separated from the growing cell by~$C$ (i.e., $H_S$ is inside~$C$).
This means that the crossing is on the convex hull boundary of a
subset of~$P$ containing $C$, a contradiction to
Observation~\ref{obs:compact_cycle_convex} and the assumption that
$C$ is a compact cycle in $\Gamma$.

\begin{figure}
\centering
\includegraphics{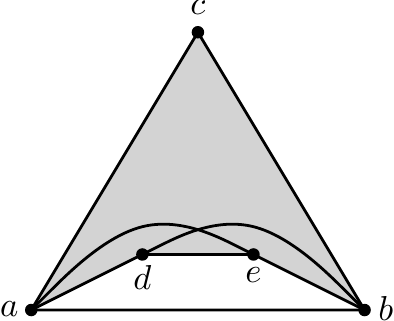}
\caption{Illustration of the proof of 
Lemma~\ref{lem:important_triangles_not_crossed} where 
$H_S=\seq{a,b,c}$ and $C=\seq{a,d,e,b}$. The shaded region is 
to the inside of both $H_S$ and $C$. (Note the direction of 
the cycles by which, e.g., $c$ is to the inside of $C$.) After removing $c$, the crossing will be on 
the unbounded face.}
\label{fig_important_triangles_not_crossed}
\end{figure}
\end{proof}

Lemma~\ref{lem:important_triangles_not_crossed} is closely 
related to Lemma~\ref{lem:important_partition} (see also 
Theorem~3.2 in Balko et al.\cite{monotone}).
\subsection{Abstract Order Types Determined by Convex 
Hull Edges}\label{sec:rotation_and_hull_edge}

Consider a radial system~$R$, and let $\chi$ be an
abstract order type that realizes $R$. Let $ab$ be a directed edge 
on the convex hull of $\chi$ so that all other points of $\chi$ 
are to the left of~$ab$.
It is easy to see that the edge $ab$ and $R$ together uniquely 
determine the convex hull of our abstract order type.
Hence, there is only one abstract order type realizing~$R$ where
$ab$ has this property.
Wismath's approach\cite{wismath} obtains a point set from 
its local sequences and its $x$-coordinates.
However, it can be observed that an abstract order type of a 
radial system is already determined by the relative horizontal 
order of the points, and the actual values of the $x$-coordinates 
are irrelevant (given that the abstract order type is the order 
type of a point set).\footnote{Wismath actually constructs 
a point set. Without being given the $x$-coordinates, deciding 
whether there exists such a point set would be equivalent to 
deciding stretchability of a pseudo-line 
arrangement\cite{semispaces}, an $\exists \mathbb{R}$-complete 
problem\cite{mnev} (see also Schaefer\cite{schaefer}).
But recall that this work is concerned only with abstract order types.}
When thinking of $a$ as a point at vertical infinity, the radial 
ordering around~$a$ gives a horizontal order of the remaining points.
We re-state the following well-known fact.

\begin{lemma}\label{lem:sidedness}
Given a radial system $R$ and a directed convex hull edge $ab$ of 
an abstract order type $\chi$ realizing $R$, 
we can compute for each triple $(i,j,k) \in [n]^3$
the orientation $\chi(i,j,k)$ 
in $O(1)$ triple queries and hence in constant time.
\end{lemma}

\begin{proof}
We know that every point except for $a$ and $b$ is to the left of $ab$.
Let $\seq{v_2 = b, \dots, v_n}$ be the linear order obtained from 
$R(a)$, starting with~$b$.
If $a$ is involved in a sidedness query, this order already 
determines the orientation of a triple.
Otherwise, let $v_i, v_j, v_k$ be a triple of points, with $i < j < k$.
Then, if $v_j$ is contained in the triangle $a v_i v_k$, the 
triple $(v_i,v_j,v_k)$ is oriented clockwise; otherwise, it is 
oriented counterclockwise. This can be checked with $O(1)$ triple
queries to $R$.
In the radial ordering around~$v_i$, this corresponds to $v_j$
being between $v_k$ and $a$, or $v_k$ being between $v_j$ and $a$.
\end{proof}

\subsection{Obtaining Hull Edges}\label{sec:obtaining_hull_edges}

Let $P$ be a generalized configuration of $n$ points, and let 
$R$ be the radial system of the abstract order type~$\chi$ of~$P$.
The goal is to find a set of $O(n)$ \emph{candidate edges} that may 
appear on the convex hull of a realization of $R$ (i.e., the edges 
of the convex hull of~$P$, if there is no other realization of~$R$, 
or the union of the edges of all important triangles).
Our algorithm incrementally builds a ``hull structure'' 
(defined below) for $P$.
Before step $k$, we have a current set 
$P_{k-1} \subseteq P$ of $k-1$ points and a hull structure 
$Z_{k-1}$ that represents the candidate edges for $P_{k-1}$.
The algorithm selects a point $p_{k} \in P \setminus P_{k-1}$, adds 
it to $P_{k-1}$, and updates $Z_{k-1}$.
A careful choice of $p_{k}$ allows for updates in 
constant amortized time.

We begin with the description of the hull structure.
Let $P_k \subseteq P$ be a set of $k$ points $(k \geq 4)$.
The $k$th \emph{hull structure} $Z_k$ is an abstract representation 
of a graph with vertex set $V_k \subseteq P_k$ that is embedded on 
the sphere. That is, $Z_k$ stores the incidences between the 
vertices, edges, and faces, but it does not assign coordinates to 
the points.
Hull structures come in three types (see \fig{fig_types}), which 
correspond in one-to-one-fashion to the three possible configurations 
of important sets in Theorem~\ref{thm:important_triangles_structure}:

\noindent
\textbf{Type~1:}
$Z_k$ is a compact cycle (recall that, therefore, $R$ 
restricted to $V_k$ represents a convex $|V_k|$-gon with $|V_k| \geq 4$).

\noindent
\textbf{Type~2:}
$Z_k$ consists of a compact cycle $C$ and a \emph{top vertex} $t$ 
that covers $C$.
The $3$-cycles incident to $t$ are called \emph{top triangles}.
A top triangle $\tau$ is marked either \emph{unexamined}, 
\emph{dirty}, or \emph{empty}.
Initially, $\tau$ is unexamined.
Later, $\tau$ is marked either dirty or empty.
``Dirty'' indicates that $\tau$ cannot contain a convex hull 
vertex in its interior.
``Empty'' means that $\tau$ is a candidate for an 
important triangle.  We orient each top triangle so that 
all other vertices of $Z_k$ are to the exterior.

\noindent
\textbf{Type~3:}
$Z_k$ is the union of two vertex-disjoint $3$-cycles $T_1$ 
and $T_2$, called \emph{independent triangles}.
$T_1$ and $T_2$ are directed so that each has all of $P_k$ 
to the interior. Moreover, the edges between the vertices of 
$T_1$ and $T_2$ appear as in \fig{fig_types}.

Let $R_k$ be the restriction of $R$ to $P_k$.
We maintain the following invariant:
(a) if $R_k$ has exactly one important set of size at least four, $Z_k$ 
is of Type~1 and represents the counterclockwise convex hull boundary;
(b) if $R_k$ has two disjoint important triangles, $Z_k$ is of Type~3, 
and the important triangles are exactly the independent triangles;
(c) if $R_k$ has several important triangles with a common vertex, $Z_k$ 
is of Type~2 and all important triangles appear as top triangles;
(d) if $R_k$ has exactly one important triangle, $Z_k$ is of Type~2 or 3, 
with the important triangle as a top triangle (Type~2) or as an 
independent triangle (Type~3).
Furthermore, if $Z_k$ is of Type~2, no convex hull vertex for $P$ 
lies inside a dirty triangle, and each point of $P_k$ lies either in 
$C$ or in a dirty triangle.

\begin{figure}[t]
\centering
\includegraphics[scale=0.8]{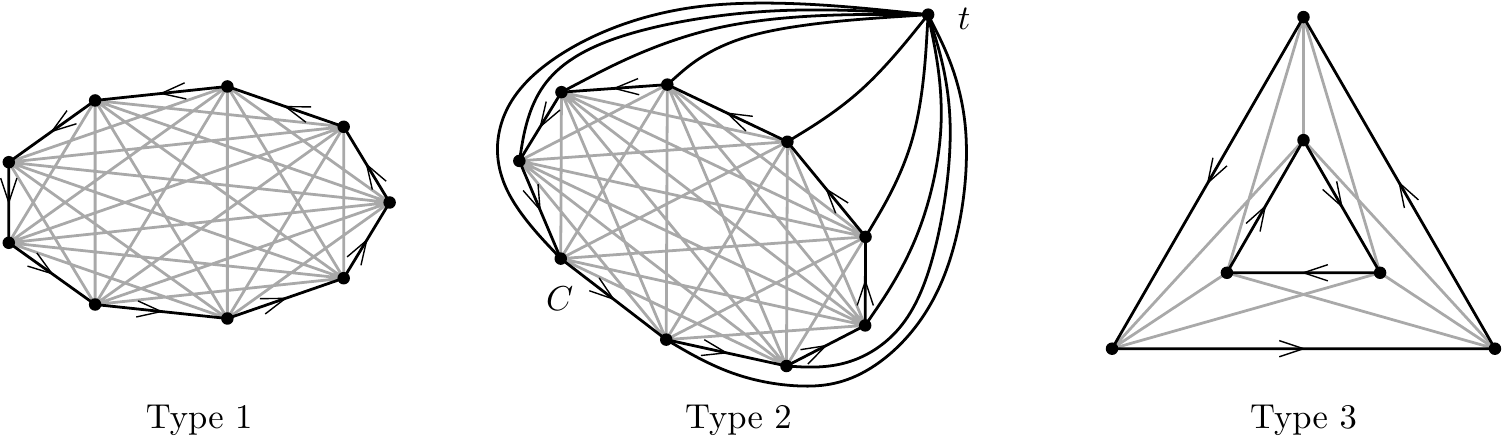}
\caption{The three different types of hull structures.}
\label{fig_types}
\end{figure}

Initially, we pick $5$ arbitrary points from $P$.
Among those, there must be a compact $4$-cycle $Z_4$ (see, e.g., 
Figure~4 in Aichholzer et al.\cite{urs}), which can be found in constant time.
Our initial hull structure $Z_4$ is of Type~1, with vertex set $V_4 = P_4$.
We next describe the insertion step for each possible type.
For the running time analysis, we subdivide the algorithm into 
\emph{phases}.
Each phase is of Type~1, 2, or 3, and a new phase begins each time 
the type of the hull structure changes.

\subsubsection{Type~1}
We pick an arbitrary vertex $c$ of $Z_{k-1}$, and we inspect
$R(c)$ to determine in 
constant time whether $c$ has an incident edge emanating 
to the outside of $Z_{k-1}$.
If not, the edges incident to $c$ in $Z_{k-1}$ are on the convex 
hull of $P$, and we are done; see below.
Otherwise, let $p_{k} \in P \setminus{P_{k-1}}$ be the endpoint 
of such an edge. We set $P_{k} = P_{k-1} \cup \{p_{k}\}$, and 
we walk along $Z_{k-1}$ (starting at $c$) to find the interval $I$ 
of vertices for which the edge to $p_{k}$ emanates to the outside
(this can be checked in $O(1)$ triple queries).
There are two cases: (i) if $I = Z_{k-1}$ (i.e., $p_{k}$ covers 
$Z_{k-1}$), then $Z_{k}$ is the hull structure of Type~2 with 
compact cycle $Z_{k-1}$, top vertex $p_{k}$, and all top 
triangles marked unexamined; (ii) if $I = \seq{c_i, \dots, c_j}$ 
is a proper subinterval of $Z_{k-1}$, the next hull structure 
$Z_{k}$ is of Type~1 with vertex sequence 
$\seq{p_k, c_j, \dots, c_i}$ (since $R$ is realizable, we
have $c_j \neq c_i$). 
\begin{lemma}\label{lem:inv_type1}
We either obtain an edge from which the convex hull of $P$ can be 
determined uniquely, or $Z_k$ is a valid hull structure for $P_k$.
\end{lemma}
\begin{proof}
By the invariant, $Z_{k-1}$ is the convex hull of~$P_{k-1}$.
If $c$ has no incident edges emanating to the outside, the edges 
$e$ and $f$ incident to $c$ in $Z_{k-1}$ lie on the convex hull of 
$P$, as $e$ and $f$ lie on the convex hull of $P_{k-1}$ and all 
points of $P$ lie in the wedge spanned by them.
Furthermore, $P_{k-1}$, and hence $P$, has at least one point 
outside the triangle spanned by $e$ and $f$, so $P$ has a 
convex hull with at least four vertices.
Then the convex hull of $P$ is unique, by
Theorem~\ref{thm:important_triangles_structure}.

Otherwise, if $p_{k}$ does not cover $Z_{k-1}$, the lemma 
follows from simple geometry. If $p_{k}$ covers $Z_{k-1}$, 
then $p_k$ must lie on every convex hull of $P_k$.
Moreover, a candidate edge of $P_k$ is either a candidate 
edge of $P_{k-1}$ or connects $p_k$ to an extreme point of $P_{k-1}$.
Thus, all important sets of $P_k$ are top triangles of $Z_{k}$.
\end{proof}

\begin{lemma}\label{lem:time_type1}
A Type~1 phase that begins with a hull structure of size $m$ 
and lasts for $\ell$ insertions takes $O(m + \ell)$ time.
Furthermore, the next phase (if any) is of Type~2, beginning 
with a hull structure of size at most $m + \ell$.
\end{lemma}

\begin{proof}
Once a vertex is deleted from the hull structure, it does not reappear, 
so the total number of distinct vertices is at most $m + \ell$.
The time to insert non-covering vertices can be charged to the deleted 
vertices. It takes $O(m + \ell + m)$ time to identify a covering 
vertex, but then the phase is over.
\end{proof}

\subsubsection{Type~2}
We begin with a simple observation.

\begin{observation}\label{obs_covering_order}
Let $Z_{k-1}$ be a Type~2 hull structure with compact 
cycle $C$ and top vertex $t$. The vertices of $C$ appear in 
their circular order in the clockwise radial ordering around $t$.
\end{observation}

We need to identify a suitable vertex $p_k$ to insert.
For this, we select an unexamined top triangle 
$\tau = \seq{t, c_{i+1}, c_i}$, and we test whether $c_i$ has 
an incident edge that emanates to the inside of $\tau$.
If yes, let $v \in P \setminus P_{k-1}$ be an endpoint of 
such an edge and check whether $c_i v$ crosses the edge $t c_{i+1}$.
If so, the 
vertices of $\tau$ lie inside a convex quadrilateral, and
by Lemma~\ref{lem:important_triangles_not_crossed} there 
is no convex hull vertex inside~$\tau$.
We mark $\tau$ dirty and proceed to the next unexamined triangle.
If not, we set $p_{k} = v$ and $P_k = P_{k-1} \cup \{ p_k \}$.
If $c_i$ has no incident edge emanating to the inside of $\tau$, 
we perform the analogous steps on $c_{i+1}$.
If $c_{i+1}$ also has no such incident edge, we mark $\tau$ empty
and proceed to the next unexamined triangle.
(The empty triangle $\tau$ might still be crossed by an
edge incident to $t$.)

\begin{lemma}\label{lem:new_p_k_type2}
We either find a new vertex $p_k$, or all candidate edges
for $P$ lie in~$Z_{k-1}$. Furthermore, no dirty triangle contains
a possible convex hull vertex of~$P$.
\end{lemma}

\begin{proof}
All top triangles that are newly marked dirty are part 
of a compact $4$-cycle, so no dirty top triangle can contain
a possible convex hull vertex of $P$ in its interior.
Suppose we fail to find a vertex $p_k$ and there is a
candidate edge $e$ for $P$ not in $Z_{k-1}$.
Then $e$ must have one endpoint $v \in P \setminus P_{k-1}$, 
since otherwise $e$ would be a candidate edge for $P_{k-1}$ 
and part of $Z_{k-1}$, by the invariant.
Now, $v$ cannot lie in $C$ or in a dirty triangle, by the invariant.
Also, $v$ cannot lie in an empty triangle, since this would have 
been detected by the algorithm.
Thus, $v$ must hide in an unexamined triangle, but there 
are no such triangles left.
\end{proof}

\begin{figure}[t]
\centering
\includegraphics[width=\textwidth]{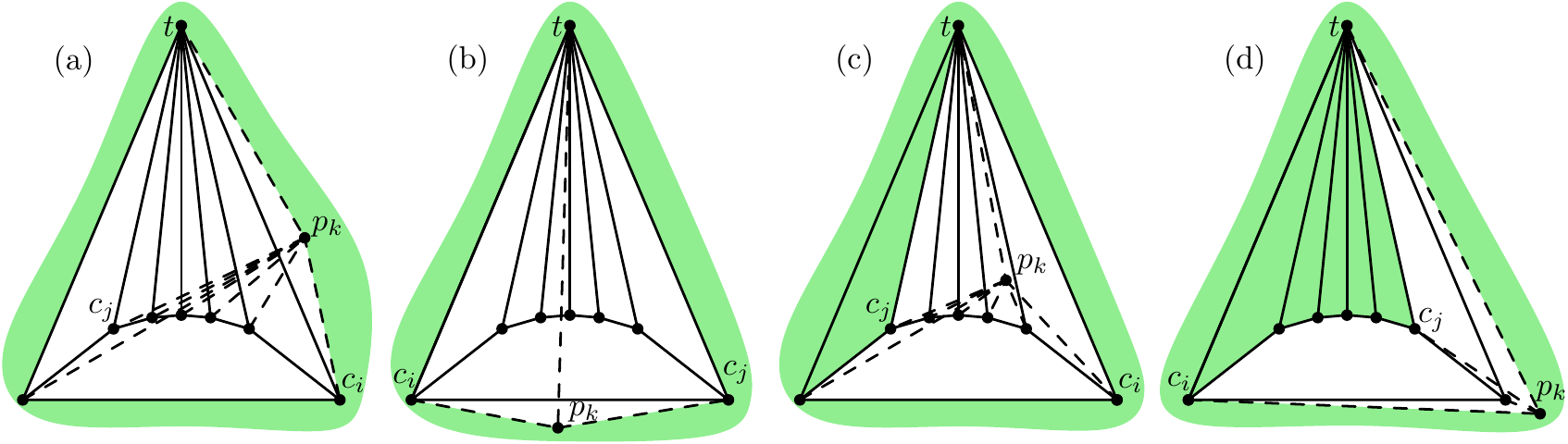}
\caption{$Z_{k-1}$ is of Type~2 and~$p_k$ is not covering: 
if~$p_k$ forms a non-crossed $4$-cycle, $Z_k$ is of Type~1~(a, b);
if not, $Z_k$ is of Type~2 with $p_k$ on the compact cycle~(c, d).
The algorithm will later discover that the triangle 
$\seq{t, c_{j+1}, c_j}$ in (c) is not important since it is 
inside a convex quadrilateral.}
\label{fig_type_2_not_covering}
\end{figure}

\begin{figure}
\centering
\includegraphics{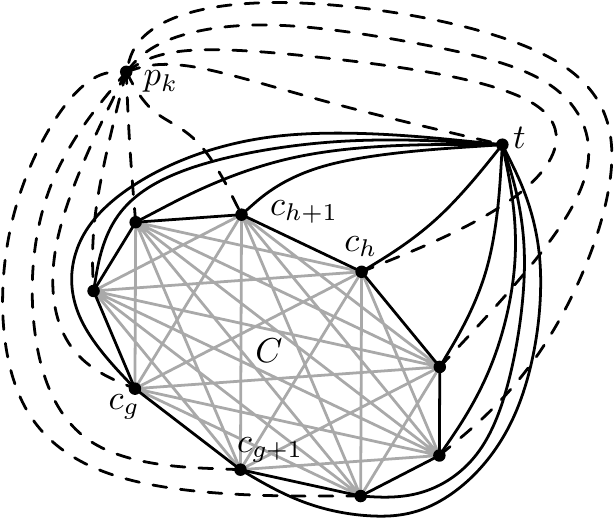}
\caption{$Z_{k-1}$ is of Type~2 and $p_k$ is covering.}
\label{fig_type_2_cover}
\end{figure}

With $p_k$ at hand, we inspect the boundary of $C$ to find 
the interval $I$ of vertices for which the edge to $p_k$ 
emanates to the outside of~$C$. First, if $p_k$ does not cover 
$C$, i.e., if $I = \seq{c_i, \dots, c_j}$ is a proper subinterval of 
$C$, then $p_k$ must lie between $c_{i-1}$ and $c_{j+1}$ in the 
clockwise order around~$t$, as in any realization one of the 
cases in \fig{fig_type_2_not_covering} applies.
If $p_k$ is between $c_{i-1}$ and $c_i$ or between $c_{j}$ and 
$c_{j+1}$, then either $\seq{p_k, t, c_{i-1}, c_i}$ or 
$\seq{t, p_k, c_{j}, c_{j+1}}$ is a compact $4$-cycle containing 
$P_k$, and we make it the next hull structure $Z_k$ of 
Type~1; see \fig{fig_type_2_not_covering}(a).
The green areas in the figures are the only regions where we 
might still find candidate edges.
Otherwise, if $i+1 = j$ and the edge $t p_k$ crosses $c_i c_{i+1}$, 
the compact $4$-cycle $\seq{t, c_j, p_k, c_i}$ contains 
$P_k$ and becomes the next Type~1 hull structure $Z_k$; 
see \fig{fig_type_2_not_covering}(b).
In any other case (i.e., $p_k$ lies between $c_i$ and $c_j$ 
in clockwise order around $t$ and if $i+1 = j$ then $tp_k$ 
does not cross $c_ic_{i+1}$), $Z_k$ is of Type~2 and obtained 
from $Z_{k-1}$ by removing the top triangles between $c_i$ and 
$c_j$ and adding the top triangles $\seq{t, p_k, c_i}$ and 
$\seq{t, c_j, p_k}$; see \fig{fig_type_2_not_covering}(c) and~(d).
If $c_ip_k$ intersects an edge of $Z_{k-1}$, then $\seq{t, p_k, c_i}$ 
lies in a compact $4$-cycle and is marked dirty.
Otherwise, it is marked unexamined.
We handle $\seq{t, c_j, p_k}$ similarly.

Second, suppose $p_k$ covers $C$ and let $g,h$ be so that $p_k$
is between $c_g$ and $c_{g+1}$ in clockwise order around $t$ and $t$
lies between $c_{h}$ and $c_{h+1}$ in clockwise order around $p_k$.
Observation~\ref{obs_covering_order} ensures that these edges are
well-defined; see \fig{fig_type_2_cover}.
Now there are three cases.
First, if $g = h$, then one of $\seq{c_{g}, c_{g+1}, t, p_k}$ 
or~$\seq{c_{g}, c_{g+1}, p_k, t}$ defines a compact 
$4$-cycle containing $P_k$, so $Z_k$ is of Type~1 and 
consists of this cycle; see \fig{fig_type_2_covering_same_disjoint_new}(a).
Second, if $\{g, g+1\} \cap \{h, h+1\} = \emptyset$, 
then $Z_k$ is of Type~3, with independent 
triangles $\seq{p_k, c_g, c_{g+1}}$ and 
$\seq{t, c_h, c_{h+1}}$; see 
\fig{fig_type_2_covering_same_disjoint_new}(b).
Third, suppose that $h = g+1$ or $g = h+1$, say, $h = g+1$.
Then $Z_k$ is of Type~2, with top vertex $c_j$ and compact 
cycle $\seq{t, p_k, c_g, c_{h+1}}$.
The top triangle $\seq{c_h, c_{h+1}, c_g}$ is dirty, the 
other top triangles are unexamined;
see \fig{fig_type_2_covering_same_disjoint_new}(c--d).

\begin{figure}[t]
\centering
\includegraphics[width=\textwidth]{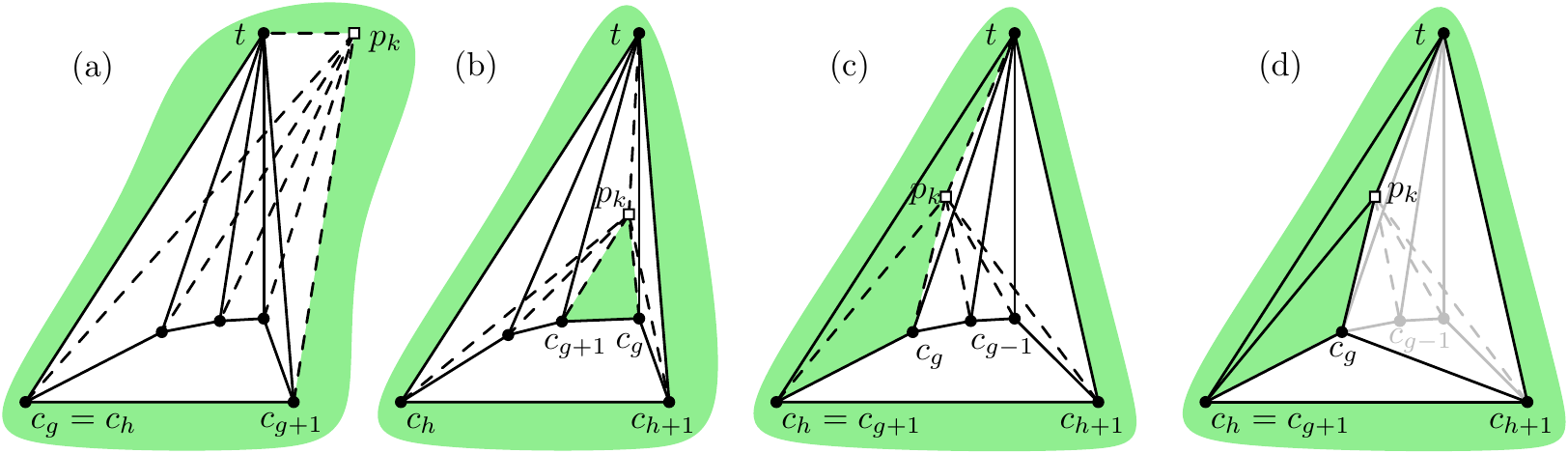}
\caption{$Z_{k-1}$ is of Type~2 and $p_k$ (box) is covering:
if $t$ and $p_k$ are between the same vertices in each 
other's rotation, $Z_k$ is of Type~1~(a);
if these vertices are disjoint, $Z_k$ is of Type~3~(b);
if~$t$ and $p_k$ have a common neighbor~$c_j$ in the other's 
rotation~(c), the new top vertex $c_h$ of $Z_k$ structure requires 
the construction of a new compact cycle~(d).
}
\label{fig_type_2_covering_same_disjoint_new}
\end{figure}

\begin{lemma}\label{lem:inv_type2}
The resulting hull structure is valid for $P_k$.
\end{lemma}
\begin{proof}
First suppose the vertices of $C$ whose edges 
to $p_k$ emanate to the outside form a proper subinterval 
$\seq{c_i, \dots, c_j}$ of $C$.
Suppose further that $p_k$ appears between $c_\ell$ 
and $c_{\ell+1}$ in the clockwise order around $t$ and consider 
some realizing abstract order type~$\chi_k$ for $P_{k}$.
In the induced realization $\chi_{k-1}$ with $P_{k-1}$, the 
top triangle $\tau = \seq{t, c_{\ell+1}, c_\ell}$ of $Z_{k-1}$ 
is either the outer face or a bounded triangle.

If $\tau$ is the outer face, suppose first that the set  
$\{t, c_{\ell+1}, c_\ell, p_k\}$ forms a compact $4$-cycle $C'$.
Since $\tau$ is the outer face of $P_{k-1}$, we know that $C'$ 
is the outer face of $P_k$.
Since $|C'|=4$, it is the only important set by 
Theorem~\ref{thm:important_triangles_structure} and we can set $Z_k=C'$.
Equivalently, our algorithm decides this case as follows.
If $C' = \seq{t, p_k, c_{\ell}, c_{\ell+1}}$, we must have $\ell = j$.
If $C' = \seq{t, c_{\ell}, c_{\ell+1}, p_k}$, we must have $\ell = i-1$.
If $C' = \seq{t, c_{\ell}, p_k, c_{\ell+1}}$, we must have 
$i = \ell$, $j = \ell+1$ and $tp_k$ crossing $c_ic_{i+1}$.
In either case, $C'$ contains $P_k$ and constitutes the 
unique convex hull of $P_k$.
We can thus set $Z_k = C'$.
Now, if $\{t, c_{\ell+1}, c_\ell, p_k\}$ does not form 
a compact $4$-cycle, the interior vertex must be $c_\ell$ or 
$c_{\ell+1}$ (as $p_k$ is on the outer face by choice of $\ell$ 
and not covering).
Thus, $\ell \in \{i, j-1\}$, and the edges of $Z_{k-1}$ 
incident to $c_{i+1}, \dots, c_{j-1}$ cannot be candidates 
(being inside a compact $4$-cycle), while the only candidate 
edges incident to $p_k$ are $p_kt$, $p_k c_i$ and $p_k c_j$.

Similarly, if $\tau$ is bounded, we must have $i \leq \ell < j$, 
and $p_k$ must form a compact cycle with $\seq{c_j, \dots, c_i}$.
The edges of $Z_{k-1}$ incident to $c_{i+1}, \dots, c_{j-1}$ are 
in a compact $4$-cycle.
The only possible candidate edges incident to $p_k$ are 
$p_kt$, $p_kc_i$, or $p_kc_j$ (the other such edges are crossed 
by edges of $Z_{k-1}$).
Thus, in the last two cases $Z_k$ is a valid Type~2 structure, 
and our algorithm covers all cases.

If $p_k$ covers~$C$, our algorithm distinguishes all 
possible rotations around $t$ and~$p_k$.
Recall that $p_k$ is inside the $3$-cycle $\seq{t, c_{g+1}, c_g}$, 
and $t$ is inside $\seq{p_k, c_{h+1}, c_h}$.

\textbf{(a)} If $g = h$, then $\{p_k, t, c_g, c_{g+1}\}$ is in 
convex position: if there were a realization with $p_k$ or $t$ 
in the interior, the rotations at $p_k$ and $t$ would be different; 
if $c_g$ or $c_{g+1}$ was in the interior, there would be a third 
vertex on $C$ showing that $p_k$ or $t$ does not cover~$C$.
Furthermore, if the edge $p_k t$ crossed $c_g c_{g+1}$, the 
rotations at $t$ and $p_k$ would be different.
Hence, either $\seq{c_g, c_{g+1}, t, p_k}$ or 
$\seq{c_g, c_{g+1}, p_k, t}$ makes a compact $4$-cycle, 
containing all of $P_k$.

\textbf{(b)} If the two 3-cycles are disjoint, then the 
edge $p_k c_h$ crosses the \emph{fan} 
$t \rightarrow \seq{c_{g+1}, \dots, c_{h-1}}$ (i.e., all edges 
from $t$ to $c_{g+1}, \dots, c_{h-1}$), and $p_k c_{h+1}$ crosses 
the fan $t \rightarrow \seq{c_{h+2}, \dots, c_g}$.
Furthermore, $t c_g$ crosses 
$p_k \rightarrow \seq{c_{h+1}, \dots, c_{g-1}}$, and 
$t c_{g+1}$ crosses $p_k \rightarrow \seq{c_{g+2}, \dots, c_h}$.
Hence, $\seq{t, c_h, c_{h+1}}$ and $\seq{p_k, c_g, c_{g+1}}$, are 
the only cells without crossed edges and the only candidates for 
the convex hull of~$P_k$. Thus, $Z_k$ is a valid Type~3 structure 
for $P_k$.

\textbf{(c)} If, w.l.o.g., $h = g+1$, then $\seq{t, p_k, c_g, c_{h+1}}$ 
forms a compact $4$-cycle, so no vertex inside it can be extremal.  
Similarly, $\seq{c_g, c_h, c_{h+1}, c_{g-1}}$ contains the triangle 
$\seq{c_g, c_h, c_{h+1}}$, which is rightfully marked dirty.
The possible candidate edges for $P_k$ are the uncrossed edges of 
$Z_{k-1}$ or the uncrossed edges between $p_{k}$ and $Z_{k-1}$, and 
they are all included in $Z_k$.
\end{proof}

\begin{lemma}\label{lem:time_type2}
A Type~2 phase that begins with a hull structure of size $m$ and lasts 
for $\ell$ insertions takes $O(m + \ell)$ time.
Furthermore, if the next phase (if any) is of Type~1, it 
begins with a hull structure of size at most~$4$.
\end{lemma}
\begin{proof}
The second claim follows by inspection.
For the first claim, note that each top triangle is marked 
dirty or empty at most once, and that each insertion creates
a constant number of new top triangles. The total running time 
for the insertion operations can be charged to marked and removed 
top triangles.
\end{proof}

\subsubsection{Type~3}

Let $T_1 = \seq{a,b,c}$ and $T_2 = \seq{a',c',b'}$ be 
the two independent triangles of $Z_{k-1}$, labeled such that the edges $aa'$, $bb'$ and $cc'$ are uncrossed in the subdrawing with these six vertices.
Further, let $p_k$ 
be an arbitrary vertex of $P \setminus P_{k-1}$.
We set $P_k = P_{k-1} \cup \{p_k\}$, and we distinguish three cases.
First, if $p_k$ is inside both $T_1$ and $T_2$, then $Z_k = Z_{k-1}$.
Second, suppose that $p_k$ is outside, say, $T_1$, and that 
$\{p_k, a, b, c\}$ forms a compact $4$-cycle $C$. (Hence, $p_k$ 
is inside $T_2$; recall that ``inside'' and ``outside'' is 
defined by the cycle's orientation.)
Then $Z_k = C$ is of Type~1.
Third, suppose that $p_k$ is outside $T_1$ but 
$\{p_k, a, b, c\}$ does not form a compact $4$-cycle.
W.l.o.g., suppose further that $a$ is inside the triangle 
$\seq{p_k, b, c}$.
There are two subcases (see \fig{fig_type_3_triangular}):
(a) if $a$ lies inside a compact $4$-cycle, we replace~$a$ 
by $p_k$ in $T_1$ to obtain an independent $3$-cycle that, 
together with $T_2$, defines $Z_k$, again of Type~3;
(b) otherwise, $a$ is an element of a compact $4$-cycle $C$ 
that involves $p_k$, one vertex of $T_2$ and one other vertex of $T_1$.
Then, $Z_k$ is a Type~2 hull structure with compact cycle $C$ 
whose top vertex is the vertex of $T_1$ that is not an element of~$C$.
The top triangles incident to the vertex of $T_2$ are marked dirty, 
the remaining top triangles are marked unexamined.

\begin{figure}[t]
\centering
\includegraphics[width=0.95\textwidth]{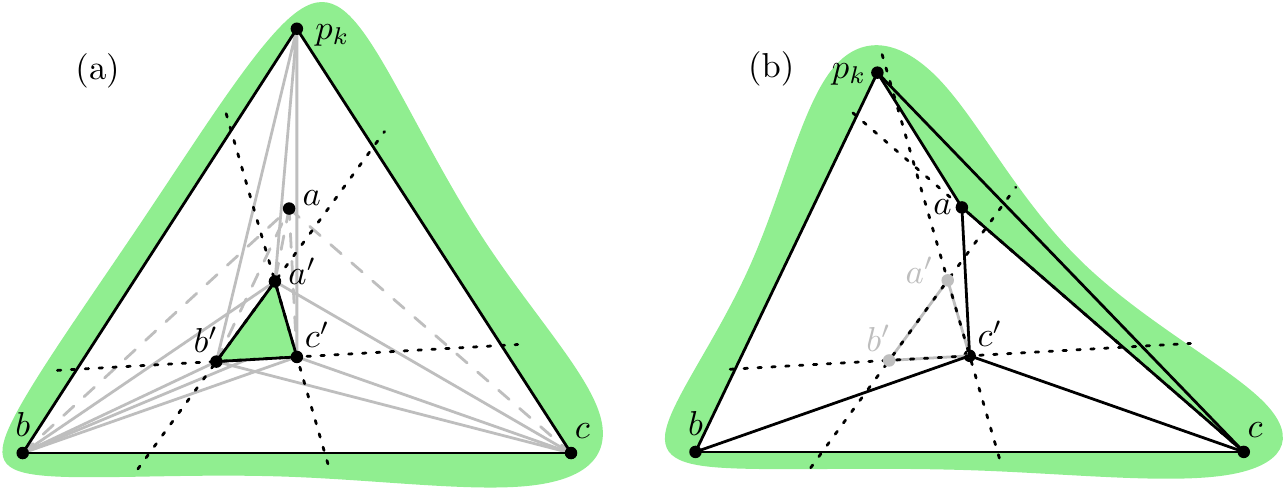}
\caption{If a vertex of an independent triangle is in a compact 
$4$-cycle (e.g., $\seq{p_k, a', c', c}$), then $Z_k$ if of Type~3~(a).
Otherwise, $Z_k$ is of Type~2 with top vertex $c$~(b).}
\label{fig_type_3_triangular}
\end{figure}

\begin{lemma}\label{lem:inv_type3}
The resulting structure $Z_k$ is a valid hull structure for $P_k$.
\end{lemma}

\begin{proof}
By the invariant, $T_1$ and $T_2$ are the only possible 
important triangles for $P_{k-1}$. Suppose that $T_1$ and $T_2$ 
are labeled such that $\seq{a,b,b',a'}$ is a compact $4$-cycle.
If $p_k$ is inside both independent triangles, then $p_k$ lies 
inside a compact $4$-cycle, and cannot have an incident candidate edge.
If $p_k$ forces four extreme points, $P_k$ lies in the corresponding 
compact $4$-cycle, and we are done.
Otherwise, let $p_k$ be outside, say, $T_1 = \seq{a,b,c}$.
Any new candidate edge for $P_k$ must be incident to $p_k$.
Furthermore, no new candidate edge is incident to $T_2$, since 
such an edge would intersect $T_1$. Thus, we have to consider the 
potential convex hulls formed by $p_k$ and the vertices of $T_1$.

Suppose first that, w.l.o.g., $a$ is contained in a compact $4$-cycle; 
see \fig{fig_type_3_triangular}(a).  This cycle must be 
$\seq{p_k, c, c', a'}$ or $\seq{p_k,b, b', a'}$.
The only possible important triangle incident to $p_k$ is 
$T_1' = \seq{p_k, b, c}$, and the $3$-cycle $\seq{p_k, b', c'}$ 
contains $a'$, so $T_1'$ and $T_2$ form the independent $3$-cycles 
of a Type~3 hull structure. There cannot be another important triangle, 
as such a triangle would have to contain a candidate edge from 
$Z_{k-1}$ and thus would be incident to $a$.

If there is no vertex of $T_1$ contained in a compact $4$-cycle, 
then an edge from $p_k$ to a vertex of $T_2$ crosses an edge of~$T_2$, 
say $a'b'$; see \fig{fig_type_3_triangular}(b).
In this case, $\seq{p_k, a, c', b}$ forms a compact $4$-cycle $C$, 
and all points from $P_k$ lie either in $C$ or in the triangles 
$\seq{c, a, c'}$ or $\seq{c, c', b}$. The latter two triangles are 
contained in the compact $4$-cycles $\seq{c, a, a', c'}$ and 
$\seq{c, c', b', b}$, and hence cannot be important.
Thus, all possible candidate edges are represented in $Z_k$.
\end{proof}

\begin{observation}\label{obs:time_type3}
A Type~3 phase with $\ell$ insertions takes $O(\ell)$ time.
If the next phase (if any) is of Type~2, it begins with a hull 
structure with at most $5$ vertices, if it is of Type~1, it begins 
with a hull structure of size $4$. 
\end{observation}

\subsubsection{Correctness and Running Time}
To wrap up, we get the following lemma:
\begin{lemma}\label{lem:wrap_up}
The final hull structure $Z_n$ contains all candidate edges for 
$R$, and it can be obtained in $O(n)$ time.
\end{lemma}
\begin{proof}
Correctness follows from Lemmas~\ref{lem:inv_type1}, \ref{lem:inv_type2}, 
and \ref{lem:inv_type3}, which show that the invariant is maintained 
throughout the construction.
By Lemmas~\ref{lem:time_type1} and \ref{lem:time_type2}, the total 
time for a phase of Type~1 and Type~2 is proportional to the number of 
insertion operations plus the initial size of the hull structure.
By Observation~\ref{obs:time_type3}, the total time for a Type-3-phase 
is proportional to the number of insertion operations.
By Lemma~\ref{lem:time_type2} and Observation~\ref{obs:time_type3}, 
every phase of Type~1 begins with a hull structure of constant size.
By Observation~\ref{obs:time_type3}, every phase of Type~2 that follows 
a phase of Type~3 has constant size.
By Lemma~\ref{lem:time_type1} and the fact that every phase of 
Type~1 begins with a structure of constant size, the size of the hull 
structure at the beginning of a phase of Type~2 that follows a phase of 
Type~1 can be charged to the number of insertions in that Type-1-phase.
The total number of insertions is $n$ (as the invariant ensures that 
in a hull structure of Type~1 or 2, every point of $P_k \setminus V_k$ 
is in a compact cycle or in a dirty triangle). 
\end{proof}

\subsection{Obtaining the Actual Hulls from a Hull Structure}

After having obtained~$Z_n$, it remains to identify the faces 
that are important sets.
If~$Z_n$ is of Type~1, then it is the only important set of~$R$.
If this is not the case, we want to obtain all the important triangles 
of~$R$, i.e., all convex hulls of abstract order types realizing 
the radial system.

\begin{lemma}\label{lem:type_2_final}
Given a Type~2 hull structure, we can decide in linear 
time which top triangles are important triangles of~$R$.
\end{lemma}

\begin{proof}
While we would only need to check top triangles that 
are not dirty, we do not use this fact in the proof.
Again, due to Lemma~\ref{lem:important_partition}, the important
triangles are exactly the top triangles
that are not contained in the interior of a compact $4$-cycle. 
Since~$t$ is an extreme point, any such compact $4$-cycle contains $t$.
Let $D = \seq{t, p, q, r}$ be such a compact 4-cycle, containing 
a top triangle $\tau = \seq{t, c_{i+1}, c_i}$.
Note that, in any abstract order type of the radial system, 
$D$ is in convex position and contains~$\tau$.
The case where the edge $c_i c_{i+1}$ is crossed by some edge 
$ts$ is evident from the radial ordering around~$t$.
But even if this is not the case, observe that there is a convex 
quadrilateral $D'$ containing $\tau$ that has either $c_i$ or 
$c_{i+1}$ as a vertex.
W.l.o.g., let $D' = \seq{t, p, q, c_i}$.

We first claim that if $D'$ exists, then there is a convex 
quadrilateral $\tilde{D} = \seq{t, \tilde{p}, \tilde{q}, c_i}$ 
such that $\tilde{D}$ contains $\tau$ and such that 
$\tilde{p}$ and $\tilde{q}$ 
are consecutive in the radial ordering around $t$. 
Such a quadrilateral clearly exists 
when $p=c_{i+1}$ and $tq$ crosses $c_i c_{i+1}$, so suppose 
this is not the case. If there is no vertex between $p$ and $q$ 
in the clockwise radial ordering around $t$, the claim is true 
with $\tilde{p} = p$ and $\tilde{q} = q$.
Otherwise, let $u$ be a vertex between $p$ and $q$.
Note that $u \neq c_{i+1}$.
If $u$ is in convex position with $t, c_i$, and $q$, we 
can replace $p$ by $u$ and obtain another quadrilateral containing~$\tau$.
Otherwise, we replace $q$ by $u$, also obtaining another convex 
quadrilateral containing $\tau$.
By this process, we eventually find $\tilde{D}$.

It now remains to show how to rule out all top triangles that 
are contained in a compact $4$-cycle.
For these, we know by the previous claim that we only need to check
consecutive pairs in the radial ordering around $t$.
Let $\seq{p_1, \dots, p_{n-1}}$ be the radial ordering around $t$ 
(with $p_1$ being an arbitrary point).
The compact cycle $C$ has exactly one edge that is intersected 
by the triangle $\seq{t, p_1, p_2}$. Let $c_j c_{j+1}$ be that edge.  
If $(c_j, c_{j+1}) = (p_2, p_1)$, then there is no vertex in
$C \setminus \{c_j, c_{j+1}\}$ that forms a compact
$4$-cycle with $\{t, c_j, c_{j+1}\}$, and we continue.
Otherwise, we remove $c_j c_{j+1}$ and 
look at its neighbors $c_{j-1} c_j$ and $c_{j+1} c_{j+2}$.
If a neighboring edge forms a compact $4$-cycle with $p_1 p_2$, we 
remove it as well. We continue this process (i.e., checking the edges of $C$ adjacent to the previously removed ones for convex position with $p_1 p_2$) until no edge is removed.
We then continue with the pair $p_2 p_3$, iteratively removing edges from $C$ that are adjacent to previously removed ones if they are in convex position with~$p_2 p_3$.
We incrementally continue this process for all pairs $p_\ell p_{\ell+1}$ and claim that the edges 
from $C$ we did not remove are exactly the ones that form an 
important triangle with~$t$.

Suppose there exists an edge $c_i c_{i+1}$ of $C$ that is 
contained in a quadrilateral formed with $p_j p_{j+1}$ but is 
not removed by this process.
The triangle $\seq{t, p_j, p_{j+1}}$ intersects the cycle~$C$ (recall
that all vertices not in $C \cup \{t\}$ lie inside $C$).
As this edge is intersected, there is a quadrilateral in convex position containing that edge.
Thus, there is a convex quadrilateral formed with all edges between 
the intersected one and $c_i c_{i+1}$, a contradiction 
(see \fig{fig_type_2_finding_hulls}).
Hence, we are left with exactly those edges that form an important 
triangle with~$t$.
\begin{figure}
\centering
\includegraphics{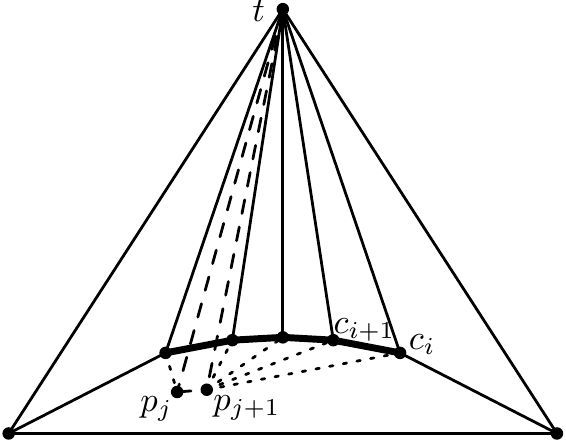}
\caption{If there is a compact $4$-cycle $\seq{t, p_j, p_{j+1}, c_i}$, 
then $c_i c_{i+1}$ is found by incrementally removing all edges of 
$C$ that form a similar quadrilateral (marked bold) with $p_j$ and 
$p_{j+1}$ when starting from the edge that is intersected by the 
triangle $\seq{t, p_j, p_{j+1}}$.
Note that this subset is realized in the same way in any 
realization of~$R$ as it is inside a compact cycle.
}
\label{fig_type_2_finding_hulls}
\end{figure}
\end{proof}

\begin{lemma}\label{lem:type_3_final}
For a Type~3 hull structure, we can decide in linear time 
which of the two independent triangles are important triangles of~$R$.
\end{lemma}

\begin{proof}
Let $T_1 = \seq{a, b, c}$ and $T_2 = \seq{a', c', b'}$ 
be the two independent triangles of the Type~3 hull structure.
(Recall that they are labeled such that the edges $aa'$, $bb'$ and $cc'$ are uncrossed in the subdrawing with these six vertices.)
We first check whether there is a partition of the vertices 
not on $T_1$ or $T_2$ into sets $P_a$, $P_b$ and $P_c$ as 
required by Lemma~\ref{lem:important_partition} with $T_2$ on 
the convex hull. We do not verify part~(3) in
Lemma~\ref{lem:important_partition} yet; hence this step is 
easily done in linear time. If such a partition does not exist, 
there is some point $p$ not within this partition and we know 
that exactly one of $T_1$ and $T_2$ is important.
Hence, $T_1$ is contained in a compact $4$-cycle~$Q$ if and only 
if $T_2$ is important.

Otherwise, suppose that such a partition does exist.
We next check the second condition in Lemma~\ref{lem:important_partition}.
For, say, $ab$, let $P_c$ be the points that are in the corresponding 
partition. We have to check whether there is a pair $(v,w)$ in $P_c$ 
that forms a compact $4$-cycle with $ab$.
We consider the elements in the clockwise radial ordering 
around $a$ between $c'$ and $c$, and compare it to the 
counterclockwise radial ordering around $b$, also between 
$c'$ and $c$. We proceed analogously for the other edges of~$T_1$.
If each pair of these orders is consistent, we know that $T_1$ 
is an important triangle under the assumption that $T_2$ is the 
convex hull of~$P$; if the assumption is not true, then $T_1$ 
has to be the convex hull of~$P$ anyway.
So suppose the orders are conflicting for, w.l.o.g., $ab$.
This means that there is a compact 4-cycle~$Q$ with $ab$ as an edge.
We know that $Q$ has to be realized by four points in convex position.
Hence, if $T_1$ is contained in~$Q$, then $T_2$ is the important triangle.
Otherwise, since $Q$ separates $T_1$ from $T_2$, $T_2$ is contained 
in $Q$ and $T_1$ is the important triangle.
\end{proof}

For each important set we obtained for the radial system~$R$, 
its chirotope is now given by Lemma~\ref{lem:sidedness}.
This proves Theorem~\ref{thm:linear_preprocessing}.

Recall that we assumed that there is at least one realization of~$R$.
We can now check this assumption in the following way.
We build the dual pseudo-line arrangement using an arbitrary 
chirotope we obtained for~$R$ using Lemma~\ref{lem:sidedness}.
This whole process takes $O(n^2)$ 
time.\cite{power_of_duality,constructing_arrangements}
If it fails then $R$ has no realization.
Otherwise, the dual pseudo-line arrangement explicitly gives the 
rotation system of the corresponding abstract order type, which 
we now compare to~$R$.

\begin{corollary}\label{cor:radial_orderings}
Testing whether a set of radial orderings is the radial 
system of an abstract order type can be done in~$O(n^2)$ time.
\end{corollary}

We can give matching lower bounds for these subtly different settings.

\begin{proposition}\label{prop:number_lower_bound}
Given a radial system~$R$ of an order type, we need $\Omega(n)$ queries
to $R$ in the worst case to determine $|T(R)|$.
\end{proposition}
\begin{proof}
We show the lower bound by an adversary argument.
Intuitively, an adversary can place any unconsidered point to 
``destroy'' the hull structure defined by the points already queried 
by an algorithm. However, as we are also given the indices of each point 
in the rotation around another one, the adversary must not place a 
point in a way that alters these indices significantly.
Thus, our proof uses the following setting.
Consider $n-1$ points in convex position and a point $t$ such that 
the hull structure of these $n$ points is of Type~2 with $t$ as 
top vertex. Let $a, b,$ and $c$ be three consecutive vertices in 
the rotation around $t$.
If the adversary moves $b$ over the edge $ac$, then only 
$R(a)$ and $R(c)$ will change:
in $R(a)$, the elements $b$ and $c$ swap their position, and 
the analogous happens in $R(c)$.
Also, note that, since the swapped elements are adjacent 
in the rotations, only the indices of these two vertices change.
So the adversary has $n-1$ points that could be moved inside the 
compact cycle of the resulting hull structure, and for each of these 
$n-1$ points, an algorithm has to determine the position in the rotation 
around one of its neighbors.
Hence, we need at least a linear number of queries.
Finally, we remark that all such abstract order types are actually 
realizable.
\end{proof}

\begin{proposition}\label{prop:existence_lower_bound}
Given a radial system~$R$, we need $\Omega(n^2)$ queries to the 
radial system to determine whether $T(R) \neq \emptyset$.
\end{proposition}

\begin{proof}
Recall Observation~\ref{obs_4tuple_swaps}.
The adversary starts by presenting the radial system of an 
(arbitrary) abstract order type. If an algorithm does not inspect 
the relative order of any adjacent pair in the rotation 
around any point, the adversary can swap exactly these pairs.
The resulting radial system cannot be the one of an abstract order 
type (and not even of a good drawing).
The quadratic lower bound follows.
\end{proof}

Note that Proposition~\ref{prop:existence_lower_bound} applies to 
checking whether, for a set $R$ of radial orderings, 
$T(R) \neq \emptyset$, while Proposition~\ref{prop:number_lower_bound} 
applies to determine $|T(R)|$ under the assumption that $|T(R)| \geq  1$, 
in the same way as Theorem~\ref{thm:linear_preprocessing} provides the 
hull structure of $R$ under the assumption that $R$ is the radial system 
of an abstract order type, while Corollary~\ref{cor:radial_orderings} is 
for checking whether a set of radial orderings is the radial system of 
an abstract order type.

We can apply our insights to obtain all important sets of a given 
chirotope.

\begin{theorem}\label{thm:other_direction}
Given an abstract order type, a hull structure of its radial 
system can be found in $O(n \log n)$ time. Further, the faces in the 
hull structure that can become convex hulls can be reported in the same 
time.
\end{theorem}

\noindent
To show this theorem, we use the following lemma.

\begin{figure}
\centering
\includegraphics{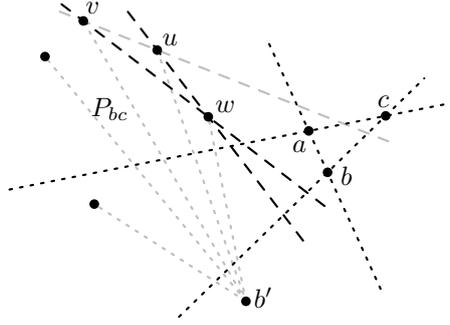}
\caption{If two points $v$ and $w$ are in convex position with 
$bc$, then for any point~$b'$ in the wedge at $b$ there is a 
consecutive pair for which this is also the case.}
\label{fig_important_partition_b}
\end{figure}

\begin{lemma}\label{lem:important_order}
Let $\seq{a, b, c}$ be an empty triangle in an abstract order 
type $\chi$ on a set~$P$. Let $P_b$ be the set of points to the right 
of $bc$ and to the right of $ab$, and let $P_{bc}$ be the set of points 
to the left of $bc$. If there exist two points $v, w \in P_{bc}$ such 
that the line $vw$ does not intersect the edge $bc$ (i.e., the four 
points are in convex position), then, for any point 
$b' \in P_b \cup \{b\}$, there are two points $v', w' \in P_{bc}$ 
that are in convex position with $bc$ and are consecutive in the 
radial ordering around~$b'$ (among the elements of $P_{bc}$).
\end{lemma}

\begin{proof}
See \fig{fig_important_partition_b}.
Observe first that if there is a point of $P_{bc}$ that is between $a$
and $b$ or $c$ (if $b' = b$) in the radial ordering around $b'$, then
$a$ and this point are in convex position with $bc$ (recall that $abc$
is empty).  We consider the linear order given by the radial ordering
of~$b'$ with $a$ as the last element.  W.l.o.g., let $v$ precede~$w$
in that linear order.  Let $u$ be a point between $v$ and $w$ (if no
such point exists, we are done).
Suppose $u$ and $bc$ are on the same side of $vw$.  Then the line (or
pseudo-line) $uw$ does not intersect the edge $bc$.  Otherwise, if $u$
and $bc$ are on different sides of $vw$, then the line $vu$ does not
intersect the edge $bc$.  Hence, this line also does not intersect
$bc$ and the two points are closer to each other in the linearized
order around~$b'$.
\end{proof}

In particular, note that if no two points of $P_{bc}$ are in convex
position with $bc$, then $P_{bc} = P_a \cup \{a\}$.

\begin{proof}[Proof of Theorem~\ref{thm:other_direction}]
For the given abstract order type on a set~$P$, construct the convex
hull $\CH(P)$ of~$P$ in $O(n \log n)$ time.%
\footnote{Knuth\cite{knuth1992axioms} discusses how to obtain the
convex hull of abstract order types in $O(n \log n)$ time.  It is also
straight-forward to adapt standard algorithms like Graham's scan.} If
it has more than three vertices, we are done.  Otherwise, let
$\seq{a,b,c}$ be the convex hull.

We first test for the case where there is another important triangle
$\seq{u,v,w}$ that does not share a vertex with~$\CH(P)$.  For this we
use Lemma~\ref{lem:important_partition}.  Radially sort the vertices
around $a$, $b$, and $c$.  Consider the clockwise order around $a$ and
the counterclockwise order around $b$, which can both be interpreted
as linear orders starting with~$c$.  The last vertex where the
prefixes of these two orders match is a vertex of $\seq{u,v,w}$, say,
$w$ for the following reason.  Let $p$ and $q$ be the first
mismatching pair (hence, $\CH(\{a,b,p,q\})$ is a quadrilateral).
Suppose first that $p$ and $q$ either precede $w$ in that order, or
$w$ is one of them.
Then $\CH(\{a,b,p,q\})$ is a quadrilateral that contains $w$, a
contradiction.  Hence, suppose both $p$ and $q$ succeed $w$ in that
order, and there is another point $r$ succeeding $w$ before one of $p$
and $q$.  Then $\CH(\{a,b,r\})$ has $u$ and $v$ in its interior, but
does not contain $w$, leading again to a contradiction.  The analogous
holds for the other pairs of extreme points.  If this method returns
three points $u$, $v$ and $w$, we can check whether the radial
orderings around $u$, $v$, and $w$ match the ones around the extreme
points for the corresponding subsets defined in
Lemma~\ref{lem:important_partition}.  We need $O(n)$ time for the
partitioning, as by Lemma~\ref{lem:important_order}, we only have to
check points that are adjacent in the radial orderings around $a$,
$b$, and~$c$.  The check at $u$, $v$, and~$w$ also takes $O(n)$ time.
If the outcome is positive, we have a valid hull structure of Type~3,
and both independent triangles can become convex hulls.

Suppose there are important triangles that share a vertex.  We guess
the covering extreme point~$a$.  (If the guess is not correct, the
following process has to be repeated at most twice for $b$ and $c$.)
The important triangles incident to~$a$ can be found in the following
way.  We obtain the radial ordering around $a$, as well as the convex
hull of $P \setminus \{a\}$.  This gives us a structure that is very
similar (if not equivalent) to a Type~2 hull structure (the ``compact
cycle'' may have only three vertices).  We can apply
Lemma~\ref{lem:type_2_final} to obtain the important triangles for
this set.
\end{proof}

\section{Minimal non-realizable Radial Systems of Arbitrary Size}\label{sec:minimal_unrealizable}

For any $k \geq 3$ we describe a radial system $R_k$ over $n = 2k + 1$
vertices which is not realizable as an abstract order type, while
every radial system induced by any strict subset of the vertices can
be realized, even as a point set order type.  This shows that
realizability of radial systems cannot be decided by checking
realizability of all induced radial systems up to any fixed constant
size.

\begin{figure}
\centering
\includegraphics[width=\textwidth]{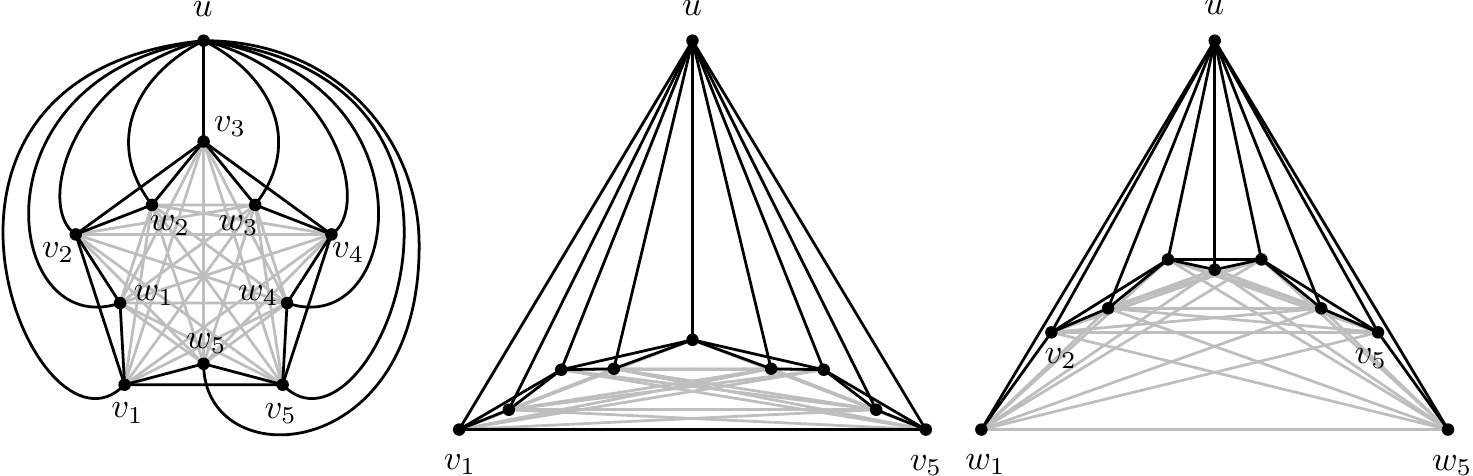}
\caption{The construction of $R_5$ on the left, 
and point set order type realizations of two induced 
radial systems after removing either $w_5$ or $v_1$ on the right.}
\label{fig_minimal_unrealizable}
\end{figure}

\begin{customthm}{\ref{thm:minimal_unrealizable}}
For any $k \geq 3$ there exists a radial system $R_k$ over $n = 2k +
1$ vertices that is not realizable as an abstract order type, while
every radial system induced by any strict subset of the vertices can
be realized as a point set order type.
\end{customthm}
\begin{proof}
Throughout, we refer to \fig{fig_minimal_unrealizable}, which
illustrates the construction of $R_5$.  We start with a so-called
\emph{double circle} with a total of $2k$ vertices.  Imagine a regular
$k$-gon with vertices $v_1,\dots,v_k$, and $k$ additional vertices
$w_1,\dots,w_k$ that are placed inside the $k$-gon and arbitrarily
close to the midpoints of its $k$ edges, where we place $w_i$ next to
the edge $v_iv_{i+1}$.  The radial system for these $2k$ vertices is
obtained by drawing all edges as straight lines and by observing the
radial orderings of the edges around each vertex.  We add one
additional vertex $u$, which can be thought of being outside of the
initial $k$-gon.  More precisely, $u$ comes directly between $v_{i-1}$
and $v_{i+1}$ in the radial ordering around any vertex $v_i$, and it
comes directly between $v_i$ and $v_{i+1}$ in the radial ordering
around any vertex $w_i$.

First, observe that edges of type $v_iw_i$ and $w_iv_{i+1}$ cannot be
boundary edges of the convex hull in any realization of $R_k$ since
they always will be in the interior of the $k$-gon $v_1,\dots,v_k$.
Second, all pairs of edges that cross in
\fig{fig_minimal_unrealizable} also cross in any other good drawing of
$R_k$, and hence they also cannot be convex hull edges.  (See the
paragraph on good drawings in Section~\ref{sec:intro}.) This leaves us
only with edges of type $uv_i$ as potential convex hull edges.
However, since there is no cycle that contains only such edges, there
is no viable candidate for the convex hull of $R_k$, which concludes
the proof of the first part of the theorem.

We now show that any strict subset of the vertices induces a radial
system that can be realized as a point set order type.  We distinguish
the following three cases.  If we remove the vertex $u$ then, by
definition, we already have an appropriate straight-line drawing of
the remaining vertices and edges.  If we remove any vertex $w_i$, say
$w_k$, then we can draw $v_1,\dots,v_k$ as a convex $k$-gon in such a
way that all edges except for $v_1v_k$ face the vertex $u$, which
means the remaining vertices $w_1,\dots,w_{k-1}$ can be added easily.
If we remove any vertex $v_i$, say $v_1$, then we reuse the drawing
from the previous case for the $(k-1)$-gon $v_2,\dots,v_k$ and the
interior vertices $w_2,\dots,w_{k-1}$.  Observe that the two
remaining vertices $w_1$ and $w_k$ do not have to be placed inside
this $(k-1)$-gon, and hence that it is simple to position them
appropriately.
\end{proof}

The above proof also works for general good drawings of 
the complete graph.

\section{Conclusion}

\begin{problem}
Can we reconstruct an order type of the vertices of a simple polygon
when given only the radial orderings of visible vertices around each
vertex (similar to Chen and Wang\cite{CW12}, but without angles)?
\end{problem}
This question is closely related to characterizing visibility graphs
of simple polygons, which is still open.  It is known that there are
infinitely many minimal forbidden induced subgraphs of visibility
graphs.  See the book of Ghosh\cite{visibility_algorithms} and
references therein.

\paragraph{Acknowledgments.}
This work was initiated during the \emph{ComPoSe Workshop on Order
Types and Rotation Systems} held in February 2015 in Strobl, Austria.
We thank the participants for valuable discussions. We would
also like to thank the anonymous reviewers for multiple suggestions
that helped to improve the presentation in the paper.

\bibliographystyle{abbrv}
\bibliography{bibliography}

\newcommand{\SortNoop}[1]{}
\begin{thebibliography}{10}

\bibitem{urs}
O.~Aichholzer, J.~Cardinal, V.~Kusters, S.~Langerman, and P.~Valtr.
\newblock Reconstructing point set order types from radial orderings.
\newblock In H.-K. Ahn and C.-S. Shin, editors, {\em Proc. 25th Annu. Internat.
  Sympos. Algorithms Comput. (ISAAC)}, volume 8889 of {\em LNCS}, pages 15--26.
  Springer-Verlag, 2014.

\bibitem{monotonicity_algorithm}
O.~Aichholzer, T.~Hackl, A.~Pilz, G.~Salazar, and B.~Vogtenhuber.
\newblock Deciding monotonicity of good drawings of the complete graph.
\newblock In {\em Proc. XVI Spanish Meeting on Computational Geometry (EGC
  2015)}, pages 33--36, 2015.

\bibitem{monotone}
M.~Balko, R.~Fulek, and J.~Kyn{\v{c}}l.
\newblock Crossing numbers and combinatorial characterization of monotone
  drawings of {$K_n$}.
\newblock {\em Discrete Comput. Geom.}, 53(1):107--143, 2015.

\bibitem{power_of_duality}
B.~Chazelle, L.~J. Guibas, and D.~T. Lee.
\newblock The power of geometric duality.
\newblock {\em BIT}, 25(1):76--90, 1985.

\bibitem{CW12}
D.~Z. Chen and H.~Wang.
\newblock An improved algorithm for reconstructing a simple polygon from its
  visibility angles.
\newblock {\em Comput. Geom. Theory Appl.}, 45(5--6):254--257, 2012.

\bibitem{DMW10}
Y.~Disser, M.~Mihal{\'a}k, and P.~Widmayer.
\newblock Reconstructing a simple polygon from its angles.
\newblock In {\em Proc. 12th Scandinavian Symp. and Workshops on Algorithm
  Theory (SWAT)}, pages 13--24, 2010.

\bibitem{constructing_arrangements}
H.~Edelsbrunner, J.~O'Rourke, and R.~Seidel.
\newblock Constructing arrangements of lines and hyperplanes with applications.
\newblock {\em SIAM J. Comput.}, 15(2):341--363, 1986.

\bibitem{felsner_weil}
S.~Felsner and H.~Weil.
\newblock Sweeps, arrangements and signotopes.
\newblock {\em Discrete Applied Mathematics}, 109(1-2):67--94, 2001.

\bibitem{visibility_algorithms}
S.~Ghosh.
\newblock {\em Visibility Algorithms in the Plane}.
\newblock Cambridge University Press, New York, NY, USA, 2007.

\bibitem{gioan}
E.~Gioan.
\newblock Complete graph drawings up to triangle mutations.
\newblock In D.~Kratsch, editor, {\em WG}, volume 3787 of {\em LNCS}, pages
  139--150. Springer, 2005.

\bibitem{goodman_proof}
J.~E. Goodman.
\newblock Proof of a conjecture of {B}urr, {G}r{\"u}nbaum, and {S}loane.
\newblock {\em Discrete Math.}, 32(1):27--35, 1980.

\bibitem{semispaces}
J.~E. Goodman and R.~Pollack.
\newblock Semispaces of configurations, cell complexes of arrangements.
\newblock {\em J. Combin. Theory Ser. A}, 37(3):257--293, 1984.

\bibitem{knuth1992axioms}
D.~E. Knuth.
\newblock {\em Axioms and Hulls}, volume 606 of {\em LNCS}.
\newblock Springer-Verlag, 1992.

\bibitem{kyncl_realizability}
J.~Kyn{\v{c}}l.
\newblock Simple realizability of complete abstract topological graphs in {P}.
\newblock {\em Discrete Comput. Geom.}, 45(3):383--399, 2011.

\bibitem{kyncl_simplified}
J.~Kyn{\v{c}}l.
\newblock Simple realizability of complete abstract topological graphs
  simplified.
\newblock In E.~D. Giacomo and A.~Lubiw, editors, {\em Graph Drawing and
  Network Visualization - 23rd International Symposium, {GD} 2015, Los Angeles,
  CA, USA, September 24-26, 2015, Revised Selected Papers}, volume 9411 of {\em
  Lecture Notes in Computer Science}, pages 309--320. Springer, 2015.

\bibitem{Matousek14}
J.~Matou\v{s}ek.
\newblock Intersection graphs of segments and $\exists\mathbb{R}$.
\newblock \texttt{arXiv:1406.2636}, 2014.

\bibitem{mnev}
N.~E. Mn{\"e}v.
\newblock The universality theorems on the classification problem of
  configuration varieties and convex polytope varieties.
\newblock In {\em Topology and Geometry---Rohlin Seminar}, volume 1346 of {\em
  Lecture Notes in Math.}, pages 527--544. Springer-Verlag, 1988.

\bibitem{ooo}
A.~Pilz and E.~Welzl.
\newblock Order on order types.
\newblock In {\em Proc. 31st Int. Sympos. Comput. Geom. (SoCG)}, pages
  285--299. LIPICS, 2015.

\bibitem{schaefer}
M.~Schaefer.
\newblock Complexity of some geometric and topological problems.
\newblock In {\em Proc. 17th Int. Symp. Graph Drawing (GD)}, pages 334--344,
  2009.

\bibitem{clusters_of_stars}
I.~Streinu.
\newblock Clusters of stars.
\newblock In {\em Proc. 23rd Annu. Sympos. Comput. Geom. (SoCG)}, pages
  439--441, 1997.

\bibitem{TFL07}
B.~Tovar, L.~Freda, and S.~M. La{V}alle.
\newblock Using a robot to learn geometric information from permutations of
  landmarks.
\newblock {\em Contemp. Math.}, 438:33--45, 2007.

\bibitem{wismath}
S.~K. Wismath.
\newblock Point and line segment reconstruction from visibility information.
\newblock {\em Internat. J. Comput. Geom. Appl.}, 10(2):189--200, 2000.

\end{thebibliography}

\end{document}